\documentclass[11pt]{article}

\usepackage[letterpaper,margin=1in]{geometry}
\usepackage{amsthm,amsmath,amsfonts,amssymb}
\usepackage[authoryear]{natbib}
\usepackage{graphicx}
\usepackage{bm}
\usepackage{booktabs}
\usepackage{array}
\usepackage{multirow}
\usepackage{adjustbox}
\usepackage{xcolor}
\usepackage{siunitx}
\usepackage{microtype}
\usepackage{hyperref}
\usepackage{url}
\hypersetup{hidelinks}

\graphicspath{{./}{images/}{GMM_PLOTS/}}

\theoremstyle{plain}
\newtheorem{theorem}{Theorem}
\newtheorem{lemma}{Lemma}
\newtheorem{proposition}{Proposition}

\theoremstyle{definition}
\newtheorem{assumption}{Assumption}

\newcolumntype{P}[1]{>{\centering\arraybackslash}p{#1}}

\newcommand{\E}{\mathbb{E}}
\newcommand{\Var}{\mathrm{Var}}

\newcommand{\red}[1]{\textcolor{black}{#1}}

\providecommand{\BigTblFont}{\scriptsize}
\providecommand{\BigTblScale}{1.00}
\providecommand{\BigTblWidth}{1.00\linewidth}

\providecommand{\DistW}{18mm}

\title{Semiparametric Regression Based on Zero Altered Models}
\author{%
Sujit K. Ghosh \quad Timothy Boakye \quad Kimberly F. Sellers\\[0.6em]
\small Department of Statistics, NC State University\\
\small \texttt{sujit\_ghosh@ncsu.edu} \quad
\texttt{tboakye@ncsu.edu} \quad
\texttt{kimberly\_sellers@ncsu.edu}}
\date{}

\begin{document}

\maketitle

\begin{abstract}
We develop a class of semiparametric zero-altered models (ZAMs) for count data, estimated through the generalized method of moments (GMM). The proposed framework extends the utility of ZAMs by allowing for both under- and over-dispersion while preserving support on the entire set of natural numbers. This offers an alternative to other models that can experience computational challenges (e.g. Conway-Maxwell-Poisson) or whose supports can be truncated as functions of the parameter space (e.g. generalized Poisson). By formulating estimation through moment restrictions rather than full likelihood specification, the GMM approach provides substantial robustness against model misspecification. In particular, inference based on likelihood methods for parametric generalized linear models (GLMs)—such as those built on Poisson, negative binomial, or related count distributions—can become unreliable when the assumed distributional form is violated, often leading to biased estimates and invalid inference. In contrast, our methodology requires only the first two moments of the distribution, ensuring consistent estimation even when higher-order distributional features are misspecified. We establish the asymptotic properties of the estimators, including consistency and asymptotic normality, under general regularity conditions. The flexibility of the approach is illustrated through both simulated experiments and real-data case studies, demonstrating superior performance relative to popular parametric GLMs, particularly in settings characterized by distributional misspecification. To facilitate practical adoption, we provide accompanying R functions for easy implementation, enabling practitioners to apply the proposed methods in applied settings without requiring bespoke programming.
\end{abstract}

\noindent\textbf{Keywords:} Count data; Generalized method of moments; Semiparametric regression; Under-dispersion; Over-dispersion

\section{Introduction}
\label{intro}

Count data arise naturally in a wide range of scientific disciplines, including epidemiology, ecology, biostatistics, insurance, and econometrics, where the response of interest represents a nonnegative integer outcome such as the number of disease occurrences, insurance claims, species counts, or equipment failures. The regression modeling of such data has traditionally been carried out within the generalized linear model (GLM) framework \citep{McCullaghNelder1989}. The classical GLM for counts is the Poisson regression model, where the response variable $Y_i$ is assumed to follow a Poisson distribution with mean $\mu_i$, linked to covariates $x_i$ through a log-link function, i.e., $\log(\mu_i) = x_i^\top \beta$. This formulation provides a parsimonious and computationally tractable framework, and has become the starting point for most count regression analyses.

However, the Poisson model has a critical limitation: it assumes that the variance of $Y_i$ equals its mean, \red{i.e.} $\text{Var}(Y_i) = \mu_i$. Empirical data frequently depart from this assumption, exhibiting either over-dispersion (variance larger than the mean) or under-dispersion (variance smaller than the mean). Over-dispersion is particularly common, arising from unobserved heterogeneity, excess zeros, or latent clustering \citep{CameronTrivedi2013, Hilbe2011}. Under-dispersion, while less frequently encountered, occurs in structured processes such as count outcomes constrained by biological or physical mechanisms  \citep{Puig2024}.

To address over-dispersion, a variety of extensions of the Poisson based GLM have been proposed. The most widely used is the negative binomial \red{(e.g., NB2)} regression model, which introduces a gamma-distributed random effect into the Poisson mean, yielding $\text{Var}(Y_i) = \mu_i + \alpha \mu_i^2$ for a dispersion parameter $\alpha > 0$ \citep{Lawless1987, CameronTrivedi1998}. Other flexible alternatives include the Poisson–inverse Gaussian \red{(PIG)} model \citep{Willmot1987}, which similarly accommodates quadratic over-dispersion, and finite mixture Poisson models \citep{Titterington1985}, which capture heterogeneity via latent mixture components.

For under-dispersed counts, several parametric distributions have been explored; \red{see \cite{Sellers2017} for details}. The \red{most notable model that allows for under-dispersion is} the Conway–Maxwell–Poisson (COM–Poisson) distribution \red{\citep{Sellers2023}}, which generalizes the Poisson by introducing a parameter $\nu$ that controls dispersion; $\nu>1$ corresponds to under-dispersion while $\nu<1$ induces over-dispersion \citep{Shmueli2005}. The COM–Poisson regression model has attracted attention due to its ability to interpolate between a wide range of dispersion structures \red{\citep{Sellers2010}}. However, underlying distribution  requires the evaluation of complex normalization constants, complicating practical use. \red{Other} models \red{that allow for under-dispersion such as} the generalized Poisson \citep{ConsulJain1973} and Poisson–binomial families \citep{HongBalakrishnan2013} \red{often entail truncation of support which likewise induces computing and analytic challenges}.


Despite their usefulness, these parametric GLM extensions suffer from several vulnerabilities. First, the likelihood-based estimation procedures inherent in GLMs require correct specification of the distributional form. When the chosen distribution fails to adequately represent the data—whether due to misspecified dispersion, incorrect zero-probability mass, or unsupported tail behavior—parameter estimates may be biased and inferential procedures can become invalid \citep{Efron1986, White1982}. Second, many under-dispersed families impose implicit truncation of the support depending on parameter values, which can complicate interpretation and practical application. Third, likelihood-based methods often entail substantial computational complexity, particularly for models requiring normalization constants (e.g. the COM–Poisson), limiting their scalability in large-sample or high-dimensional applications.

The recognition of these limitations has motivated a growing interest in semiparametric and moment-based approaches that relax stringent distributional assumptions. Generalized method of moments (GMM) estimation \citep{Hansen1982} provides a natural framework, as it relies on low-order moment conditions—typically the mean and variance—rather than full specification of the likelihood. By focusing only on the first two moments, GMM-based estimators remain consistent even when higher-order distributional features are misspecified, thereby offering robustness to model misspecification. This robustness is particularly appealing in the context of count regression, where the true generating mechanism is often complex, heterogeneous, and not easily captured by any single parametric distribution.

A particularly compelling semiparametric framework is the class of zero-altered models (ZAMs) by \citet{Ghosh2007}
. They propose a flexible semiparametric family that accommodates both under- and over‐dispersion; maintains support on the entire set of natural numbers; restricts the dispersion parameter $\delta$ to $(-1,1)$, without truncation complications; and ensures consistency and asymptotic normality of the estimator under general regularity conditions. Building on this foundation, we introduce a ZAM–GMM regression framework for count data. 
Specifically, we embed the zero-altered class into a regression setting where both the conditional mean and the dispersion parameter are linked to covariates. Estimation proceeds via GMM based solely on moment restrictions—most notably the mean and variance—offering robustness to likelihood misspecification and computational simplicity relative to parametric GLMs. We establish the asymptotic properties of the proposed estimators, including consistency and asymptotic normality, under broad regularity conditions. Through both simulation experiments and real-data applications, we demonstrate the practical advantages of our approach, highlighting its superior performance relative to popular parametric GLMs, especially under distributional misspecification. 

Recent years have seen growing interest in semiparametric and moment-based methods for discrete outcomes\red{, e.g.} semiparametric zero-inflated or augmented count models that reduce reliance on parametric assumptions \citep{He2015,Tang2015,Bertoli2021}. \red{Few models, however,} accommodate both under‐ and over‐dispersion \red{with computational ease} while preserving full support. The ZAM framework provides an extremely flexible class of models that can accommodate both under- and over-dispersion while maintaining support on the full set of natural numbers, thereby avoiding the truncation issues present in many under-dispersed parametric models. Coupling ZAMs with the GMM estimation strategy enables robust and computationally efficient inference, alleviating many of the vulnerabilities of likelihood-based GLMs. Accordingly, our approach fills \red{these} gap by combining the flexibility of ZAMs with the robustness of GMM, providing a unified, practically tractable framework.

The remainder of the paper is organized as follows. Section 2 introduces the class of \red{ZAMs} and the associated regression framework. Section 3 develops the GMM estimation procedure and theoretical properties. Section 4 presents simulation studies, while Section 5 illustrates applications to real data. Section 6 concludes with a discussion of practical implications and potential future extensions.

Technical derivations, sufficient conditions for the moment and envelope assumptions, proofs of the theoretical results, notation, the two step GMM algorithm, and simulation data generating distributions are collected in Appendices~\ref{app:quadratic_variance_derivation}--\ref{app:simulation_dgps}.

\section{Regression Framework for Zero--Altered Models}
\label{zamreg}

The zero--altered models (ZAMs) introduced by \citet{Ghosh2007} provide a semiparametric framework for inference with count data that extends beyond the limitations of standard parametric families. Unlike Poisson, negative binomial, or COM--Poisson regression models, which require explicit parametric specification of the probability mass function, ZAMs are defined through their first two moments and maintain support on the entire set of natural numbers $\mathbb{N}=\{0,1,2,\dots\}$. This feature is particularly valuable for modeling under--dispersed data, where many parametric models implicitly truncate the support depending on parameter values.

\subsection{Zero--altered model framework}

Let $f_0$ be an equi-dispersed baseline probability mass function (pmf) on $\mathbb{N}$ with
\[
\mu_0 := \mathbb{E}_{f_0}[U],
\qquad
\mathrm{Var}_{f_0}(U)=\mu_0,
\qquad
\pi_0 := f_0(0).
\]

For a dispersion index $\delta\in(-1,1)$, define the zero--altered pmf (ZA transform) \citep[Eq.~(1)]{Ghosh2007}
\begin{equation}\label{eq:ZApmf}
f_\delta(y)=
\begin{cases}
\delta_+^{\,2} + (1-\delta^{2})\,\pi_0, & y=0,\\[4pt]
\omega(\delta,\pi_0)\, f_0(y), & y\ge 1,
\end{cases}
\qquad
\delta_+=\max\{\delta,0\},\quad
\delta_-=\max\{-\delta,0\},
\end{equation}
where the mass preserving weight on the positive counts is
\begin{equation}\label{eq:omega}
\omega(\delta,\pi_0)
=
1-\delta_+^{\,2}+\delta_-^{\,2}\frac{\pi_0}{1-\pi_0},
\qquad
\omega(\delta,\pi_0)>0.
\end{equation}

Under \eqref{eq:ZApmf}, the mean and variance satisfy
\begin{equation}\label{eq:ZA_moments}
\mu=\mathbb{E}_{f_\delta}[Y]=\omega\mu_0,
\qquad
\mathrm{Var}_{f_\delta}(Y)=\omega\mu_0+\frac{1-\omega}{\omega}\mu^2.
\end{equation}
Since $f_0$ is assumed equi-dispersed, that is, $\mathrm{Var}_{f_0}(U)=\mu_0$, \eqref{eq:ZA_moments} can be written in the quadratic variance form
\begin{equation}\label{eq:quadvar}
\mathrm{Var}_{f_\delta}(Y)=\mu+\phi\mu^2,
\qquad
\phi=\frac{1-\omega}{\omega}\in(-1,\infty).
\end{equation}
Here, $\phi<0$ corresponds to under--dispersion, $\phi=0$ to equi-dispersion , and $\phi>0$ to over--dispersion. Equivalently, under the equi-dispersed baseline assumption, $\delta<0$ implies under--dispersion, $\delta=0$ implies equi-dispersion , and $\delta>0$ implies over--dispersion. When $\delta=0$, we have $\omega=1$ and $f_\delta=f_0$, so if the baseline $f_0$ is Poisson, the ZAM reduces to the Poisson model. The ZA transform preserves the full support of $f_0$ on $\mathbb{N}$ provided $f_0(y)>0$ for all $y\in\mathbb{N}$ and $\omega(\delta,\pi_0)>0$.

As shown by \citet{Ghosh2007}, any discrete distribution with a finite second moment can be represented in the form \eqref{eq:ZApmf} for suitable choices of the baseline distribution \(f_0\), its zero probability \(\pi_0\), and the dispersion index \(\delta\). Moreover, a GMM approach allows \(\delta\) to be estimated consistently without explicitly estimating \(f_0\). Thus, \(f_0\) can be treated as an infinite-dimensional nuisance parameter, which gives the ZAM framework its semiparametric nature.

\subsection{Regression Models Framework}

To extend ZAMs to a regression framework, we model both the conditional mean and the dispersion index as functions of covariates. Specifically, let $Y_i$ denote the response for subject $i=1,\ldots,n$, and let $X_i\in\mathbb{R}^p$ be a $p\times 1$ vector of covariates. Let
\[
\theta=(\beta^\top,\gamma^\top)^\top,
\qquad
\beta,\gamma\in\mathbb{R}^p.
\]

Define the inverse link functions for the mean and dispersion parameters,
\begin{equation}\label{eq:links_reg}
\mu_i(\beta)=g(X_i^\top\beta),
\qquad
1+\phi_i(\gamma)=h(X_i^\top\gamma),
\end{equation}
where $g,h:(-\infty,\infty)\to(0,\infty)$ are strictly increasing. Hence,
\[
\phi_i(\gamma)>-1
\qquad
\text{for all } i \text{ and any } \gamma\in\mathbb{R}^p.
\]

A convenient choice is $g(u)=\exp(u)$ and $h(u)=\exp(u)$, which gives
\[
\mu_i=\exp(X_i^\top\beta),
\qquad
\phi_i=\exp(X_i^\top\gamma)-1\in(-1,\infty).
\]

Combining \eqref{eq:links_reg} with \eqref{eq:quadvar} gives the conditional variance
\begin{equation}\label{eq:conditional_variance}
\mathrm{Var}(Y_i\mid X_i)
=
\mu_i+\phi_i\mu_i^2,
\qquad
\phi_i=h(X_i^\top\gamma)-1.
\end{equation}
This formulation allows covariates to affect dispersion, which is not available in the standard Poisson GLM or various extension to over-dispersed models.

\subsection{Relation to GLMs}

The proposed ZAM regression can be viewed as a semiparametric extension of GLMs:
\begin{enumerate}
    \item The mean structure follows the usual GLM specification, linking $\mu_i$ to covariates through the inverse mean function $g(\cdot)$.

    \item The variance is not restricted to equal the mean, nor to follow a fully specified parametric family. Instead, it arises from the ZA moment structure and is modeled as
    \begin{equation*}
    \mathrm{Var}(Y_i\mid X_i)=\mu_i+\phi_i\mu_i^2,
    \end{equation*}
    which accommodates both under-dispersion and over-dispersion depending on $\phi_i$, while leaving the baseline pmf otherwise unspecified, provided the baseline $f_0$ is equi-dispersed.
\end{enumerate}

In contrast to likelihood based GLMs, estimation is carried out using GMM. This approach uses the ZAM moment conditions for $\mathbb{E}(Y_i)$ and $\mathrm{Var}(Y_i)$, or equivalently the factorial moment $\mathbb{E}[Y_i(Y_i-1)]$. Thus, the estimating equations rely only on the first two conditional moments, although standard errors and asymptotic normality require additional regularity conditions, such as suitable finite higher order moments. The result is a robust semiparametric regression framework that is less sensitive to full likelihood misspecification than parametric GLMs, provided the stated moment conditions are correctly specified.

\subsection{Advantages of the ZAM regression approach}

The ZAM regression framework offers several advantages:
\begin{enumerate}
    \item[(i)] \textbf{Flexibility in dispersion:} Both over-dispersion and under-dispersion are naturally accommodated, with dispersion linked to covariates via \eqref{eq:links_reg}.

    \item[(ii)] \textbf{Full support:} Provided the baseline $f_0$ has full support on $\mathbb{N}$ and $\omega(\delta,\pi_0)>0$, the ZA pmf \eqref{eq:ZApmf} preserves that support for all admissible $\delta$. At the moment level, the variance model remains well defined for all $\phi>-1$. This avoids the truncation issues of many under-dispersed parametric families.

    \item[(iii)] \textbf{Semiparametric robustness:} Inference relies only on low order moment conditions, namely the mean and second factorial moment, which helps protect against likelihood misspecification.

    \item[(iv)] \textbf{Generality:} Classical GLMs, such as Poisson regression, are nested as special cases by taking $\phi_i\equiv 0$, so the framework unifies parametric and semiparametric count regression.

    \item[(v)] \textbf{Relation to quasi likelihood methods:}
Like quasi Poisson and other quasi likelihood approaches, the proposed method does not require a fully specified probability distribution. However, the ZAM framework models dispersion directly as a function of covariates, allowing the amount and direction of dispersion to vary across observations. This gives a more flexible interpretation of under-dispersion, equi-dispersion , and over-dispersion within one framework \citep{Wedderburn1974,Smyth1989}.
\end{enumerate}

Thus, regression models based on ZAMs provide a natural extension of GLMs, with greater robustness and flexibility for count data analysis.

\section{Estimation via the Generalized Method of Moments}
We now develop the estimation procedure for regression models based on zero altered models (ZAMs) using the generalized method of moments (GMM). The GMM approach provides a semiparametric estimation strategy that relies only on the low order moment conditions implied by the ZAM framework, and so avoids the need for full distributional specification. This yields estimators that are both flexible and robust to full likelihood misspecification when the moment conditions are correctly specified. In addition, the ZAM framework gives the dispersion regression a direct interpretation through the underlying ZA parameters $\omega$ and $\delta$, which describe how probability mass is altered at zero and across the positive counts.

\subsection{Moment conditions}

Recall that, for observation $i$ with covariate vector $X_i$, the conditional mean and dispersion factor are defined in \eqref{eq:links_reg}, and the conditional variance is given in \eqref{eq:conditional_variance}.

Using the factorial moment identity,
\[
\mathbb{E}[Y_i(Y_i-1)\mid X_i]
=
\mathrm{Var}(Y_i\mid X_i)+\mu_i^2-\mu_i,
\]
we have
\begin{equation*}
\mathbb{E}\!\big[Y_i(Y_i-1)\mid X_i\big]
=
\mu_i(\beta)^2\{1+\phi_i(\gamma)\},
\end{equation*}
where $\mu_i(\beta)$ and $\phi_i(\gamma)$ are defined in \eqref{eq:links_reg}. This identity provides the observable second factorial moment used below.

The ZAM framework implies the following conditional moment restrictions:
\begin{align}
\mathbb{E}\!\left[Y_i-\mu_i(\beta)\mid X_i\right] &= 0,
\label{eq:mom1}\\
\mathbb{E}\!\left[
Y_i(Y_i-1)-\mu_i(\beta)^2\{1+\phi_i(\gamma)\}
\mid X_i
\right] &= 0.
\label{eq:mom2}
\end{align}
where $\mu_i(\beta)$ and $\phi_i(\gamma)$ are defined in \eqref{eq:links_reg}.

For the basic implementation, we use $X_i$ itself as the instrument vector. Multiplying the two residuals by $X_i$ gives
\[
\mathbb{E}\!\big[X_i\{Y_i-\mu_i(\beta)\}\big]=0,
\]
and
\[
\mathbb{E}\!\big[
X_i\{Y_i(Y_i-1)-\mu_i(\beta)^2(1+\phi_i(\gamma))\}
\big]=0.
\]

Stacking the two scalar residuals yields, for observation $i$, the raw residual moment vector
\[
m_i(\theta)=
\begin{pmatrix}
Y_i-\mu_i(\beta)\\[4pt]
Y_i(Y_i-1)-\mu_i(\beta)^2\{1+\phi_i(\gamma)\}
\end{pmatrix},
\qquad
\mathbb{E}[m_i(\theta_0)\mid X_i]=0,
\]
where $\theta_0$ denotes the true parameter value. The corresponding instrumented GMM moment vector is
\[
\psi_i(\theta)=
\begin{pmatrix}
X_i\{Y_i-\mu_i(\beta)\}\\[4pt]
X_i\{Y_i(Y_i-1)-\mu_i(\beta)^2(1+\phi_i(\gamma))\}
\end{pmatrix}
\in\mathbb{R}^{2p},
\qquad
\mathbb{E}\{\psi_i(\theta_0)\}=0.
\]

This choice uses the same covariates in both the mean and dispersion moment equations, in line with the ZAM regression setup above. When $X_i\in\mathbb{R}^p$ and $\theta\in\mathbb{R}^{2p}$, this gives $2p$ moment conditions for $2p$ parameters, so the basic implementation is just identified \citep{Hansen1982}.

More generally, other valid functions of $X_i$ may be used as instruments. Additional valid instruments can produce a GMM system with at least as many moment conditions as parameters, so that $s\ge 2p$, with $s>2p$ corresponding to an overidentified model \citep{Chamberlain1987,Newey1990}, where $s$ denotes the dimension of the resulting GMM moment vector, or equivalently, the total number of moment conditions generated by the chosen instruments.

\subsection{Asymptotic properties and two step GMM estimation}

We now study the consistency and asymptotic normality of the GMM estimator for the ZAM regression model. Following the established large sample theory of GMM \citep{Newey1990}, the argument relies on the first two conditional factorial moments, smooth monotone link functions, finite moment conditions, identification, and full rank design assumptions. Thus, the theory applies to count regression models whose first two conditional factorial moments satisfy the ZAM moment restrictions and whose moment functions satisfy the regularity conditions below.

Throughout this section, we use $\psi(\mathcal U_i,\theta)$ for the GMM moment vector, $S_0$ for the covariance matrix of the moment vector at the true parameter value, and $W_n$ for the GMM weighting matrix.

\subsubsection*{Model, link functions, and moment restrictions}

Let
\(
\mathcal U_i := \big(Y_i,\;X_i^\top\big)^\top
\in\mathbb R^{1+p},
\)
where $X_i\in\mathbb R^p$ is the regression covariate used in both the mean and dispersion links.

Although separate design vectors could be used for the mean and dispersion parts, in the varying dispersion model we are interested in determining the effects of the available covariates on both the conditional mean and the dispersion. We therefore use the same available covariates in both $\mu_i(\beta)$ and $\phi_i(\gamma)$. The Wald tests developed later can then be used to determine which covariates significantly influence the mean and dispersion.

At the true value
\(
\theta_0=(\beta_0^\top,\gamma_0^\top)^\top,
\)
the first two conditional factorial moments satisfy
\begin{equation}\label{eq:M}
\mathbb E(Y_i\mid X_i)=\mu_i(\beta_0),
\qquad
\mathbb E\!\big[Y_i(Y_i-1)\mid X_i\big]
=
\mu_i(\beta_0)^2\{1+\phi_i(\gamma_0)\},
\end{equation}
where $\mu_i(\beta)$ and $\phi_i(\gamma)$ are defined in \eqref{eq:links_reg}.

Define the residuals
\begin{equation*}
r_{1i}(\beta)=Y_i-\mu_i(\beta),
\qquad
r_{2i}(\beta,\gamma)
=
Y_i(Y_i-1)-\mu_i(\beta)^2\{1+\phi_i(\gamma)\}.
\end{equation*}

Using $X_i$ as the instrument vector in both moment blocks, define the GMM moment function
\begin{equation}\label{eq:momentvector}
\psi(\mathcal U_i,\theta)=
\begin{bmatrix}
X_i\,r_{1i}(\beta)\\[2mm]
X_i\,r_{2i}(\beta,\gamma)
\end{bmatrix}
\in\mathbb R^{2p},
\qquad
\bar{\psi}_n(\theta)
=
\frac{1}{n}\sum_{i=1}^n
\psi(\mathcal U_i,\theta).
\end{equation}

At the true parameter value,
\[
\mathbb E\{\psi(\mathcal U_i,\theta_0)\}=0.
\]

Let $W_n$ be a positive definite weighting matrix. The GMM estimator is obtained as
\begin{equation}\label{eq:estimator}
\hat\theta_n
=
\arg\min_{\theta\in\Theta} Q_n(\theta),
\qquad
Q_n(\theta)
=
\bar\psi_n(\theta)^\top
W_n
\bar\psi_n(\theta).
\end{equation}

For later use, write
\[
\Psi(\theta)
=
\mathbb E[\psi(\mathcal U,\theta)],
\qquad
S(\theta)
=
\mathbb E[
\psi(\mathcal U,\theta)
\psi(\mathcal U,\theta)^\top
],
\qquad
S_0=S(\theta_0).
\]

For the instrument choice in \eqref{eq:momentvector}, the number of moment conditions is $s=2p$ and the number of parameters is $d=2p$. Thus, the basic implementation has $s=d$. More generally, standard GMM theory allows additional valid instruments, giving $s\ge d$, with $s>d$ corresponding to overidentification \citep{Hansen1982}.

\subsubsection*{Regularity assumptions}

For the asymptotic results below, we impose the additional smoothness condition that the inverse link functions $g$ and $h$ defined in \eqref{eq:links_reg} are twice continuously differentiable.

Throughout this subsection, $\|\cdot\|_2$ denotes the Euclidean norm. In particular, for any vector
$v=(v_1,\ldots,v_d)^\top\in\mathbb R^d$,
\[
\|v\|_2
=
\left(
\sum_{j=1}^d v_j^2
\right)^{1/2}.
\]

\begin{assumption}[Sampling]\label{ass:A1}
$\{\mathcal U_i\}_{i=1}^n$ are i.i.d.\ observations from a distribution satisfying the moment restrictions in \eqref{eq:M}.
\end{assumption}

\begin{assumption}[Moments and envelopes]\label{ass:M}
The following conditions hold.

\begin{enumerate}
    \item[(i)] There exists $\delta_X>0$ such that
    \[
    \mathbb E\|X\|_2^{2+\delta_X}<\infty.
    \]

    \item[(ii)] The response satisfies
    \[
    \mathbb E(Y^4)<\infty.
    \]

    \item[(iii)] There exist envelope functions $M_0(\mathcal U)$ and $M_1(\mathcal U)$ such that
    \[
    \sup_{\theta\in\Theta}
    \|\psi(\mathcal U,\theta)\|_2
    \le M_0(\mathcal U),
    \]
    and
    \[
    \sup_{\theta\in\Theta}
    \left\|
    \frac{\partial\psi(\mathcal U,\theta)}
    {\partial\theta^\top}
    \right\|_2
    \le M_1(\mathcal U),
    \]
    with
    \[
    \mathbb E[M_0(\mathcal U)^2]<\infty,
    \qquad
    \mathbb E[M_1(\mathcal U)]<\infty.
    \]
\end{enumerate}
\end{assumption}

\begin{assumption}[Parameter space and positivity]\label{ass:A3}
The parameter space $\Theta$ is compact and
$\theta_0\in\mathrm{int}(\Theta)$.

For all admissible $\beta$ and $\gamma$,
\[
g(X_i^\top\beta)>0,
\qquad
h(X_i^\top\gamma)>0
\]
almost surely. When uniform lower bounds are needed, we assume that there exist constants $c_g>0$ and $c_h>0$ such that
\[
\inf_{\beta\in\Theta_\beta}
g(X_i^\top\beta)\ge c_g,
\qquad
\inf_{\gamma\in\Theta_\gamma}
h(X_i^\top\gamma)\ge c_h
\]
almost surely, where \(\Theta_\beta=\pi_\beta(\Theta)\) and
\(\Theta_\gamma=\pi_\gamma(\Theta)\) denote the projections of \(\Theta\)
onto the mean and dispersion parameter spaces, respectively.
\end{assumption}

\begin{assumption}[Design]\label{ass:A4}
For all \(i=1,\ldots,n\),
\(
\mathbb E[X_iX_i^\top]\succ 0,
\)
that is, \(\mathbb E[X_iX_i^\top]\) is positive definite.
\end{assumption}

\begin{assumption}[Identification and weighting matrix]\label{ass:A5}
The population moment condition is uniquely solved at the true parameter:
\[
\Psi(\theta)=0
\quad\Longrightarrow\quad
\theta=\theta_0.
\]

Also,
\[
W_n\xrightarrow{p}W_0\succ0.
\]

For the optimal two step GMM estimator,
\[
W_0=S_0^{-1}.
\]
\end{assumption}

\paragraph*{Remark.}
Assumption~\ref{ass:M} collects the main moment conditions. The fourth moment of $Y$ is used because the second block of the moment vector contains $Y_i(Y_i-1)$. Therefore, the covariance matrix of the moment vector involves fourth order moments of $Y_i$.

\paragraph*{What the moment and envelope assumptions guarantee.}
\label{Ass2_guarantee}

The moment and envelope assumptions ensure suitable integrability of the moment functions and their derivatives. These conditions support a uniform law of large numbers for $\bar\psi_n(\theta)$, convergence of its derivative, and a multivariate central limit theorem for
$\sqrt n\,\bar\psi_n(\theta_0)$. Conditions of this type are standard in the large sample theory of GMM and conditional moment estimation \citep{Chamberlain1987,Newey1990}.

We now provide a series of technical results establishing the large sample
properties and inferential procedures for the proposed GMM estimator.

\subsubsection*{Derivatives and block Jacobian}

Throughout this section, derivatives of scalar residuals with respect to $\beta^\top$ and $\gamma^\top$ are written as row vectors.

\subsubsection*{Population Jacobian}

\begin{lemma}[Jacobian]\label{lem:jacobian}
Under Assumptions~\ref{ass:M}--\ref{ass:A4}, the population Jacobian
\[
G
=
\left.
\mathbb E\!\left[
\frac{\partial\psi(\mathcal U,\theta)}
{\partial\theta^\top}
\right]
\right|_{\theta=\theta_0}
\]
has the block form
\begin{equation}\label{eq:Gblock}
G
=
-\begin{bmatrix}
\mathbb E[
Xg'(X^\top\beta_0)X^\top
]
&
0
\\[2mm]
\mathbb E[
2Xg(X^\top\beta_0)
g'(X^\top\beta_0)
h(X^\top\gamma_0)X^\top
]
&
\mathbb E[
Xg(X^\top\beta_0)^2
h'(X^\top\gamma_0)X^\top
]
\end{bmatrix}.
\end{equation}
\end{lemma}

\begin{lemma}[Identification]\label{lem:ident}
Under Assumptions~\ref{ass:M} and \ref{ass:A4}, together with the strict monotonicity of $g$ and $h$, the population moment condition satisfies
\[
\Psi(\theta)=0
\quad\text{if and only if}\quad
\theta=\theta_0.
\]
\end{lemma}

\begin{lemma}[Uniform law]\label{lem:ULLN}
Under Assumptions~\ref{ass:A1}--\ref{ass:A3},
\[
\sup_{\theta\in\Theta}
\|\bar{\psi}_n(\theta)-\Psi(\theta)\|
\xrightarrow{p}
0.
\]
\end{lemma}

\begin{theorem}[Consistency]\label{thm:consistency}
Under Assumptions~\ref{ass:A1}--\ref{ass:A5},
\[
\hat\theta_n\xrightarrow{p}\theta_0.
\]
\end{theorem}

\subsubsection*{Asymptotic normality}

Let
\(
S_0
=
\mathbb E[
\psi(\mathcal U,\theta_0)
\psi(\mathcal U,\theta_0)^\top
].
\)

\begin{lemma}[Linearization]\label{lem:linearization}
Under Assumptions~\ref{ass:A1}--\ref{ass:A5},
\[
\sqrt n(\hat\theta_n-\theta_0)
=
-\big(G^\top W_0G\big)^{-1}
G^\top W_0
\sqrt n\,\bar{\psi}_n(\theta_0)
+o_p(1),
\]
where $G$ is defined in \eqref{eq:Gblock}.
\end{lemma}

\begin{lemma}[Full rank]\label{lem:fullrank}
Under Assumptions~\ref{ass:A3}--\ref{ass:A4}, the matrix $G$ has full column rank $2p$.
\end{lemma}

\begin{theorem}[Asymptotic normality]\label{thm:AN}
Under Assumptions~\ref{ass:A1}--\ref{ass:A5},
\[
\sqrt n(\hat\theta_n-\theta_0)
\Rightarrow
\mathcal N(0,V),
\]
where
\[
V
=
(G^\top W_0G)^{-1}
G^\top W_0S_0W_0G
(G^\top W_0G)^{-1}.
\]

If $W_0=S_0^{-1}$, corresponding to the optimal two step GMM weighting matrix, then
\[
V
=
(G^\top S_0^{-1}G)^{-1}.
\]
\end{theorem}

\subsection{Two step efficient GMM estimation}

The optimal weighting matrix $S_0^{-1}$ is unknown and must be estimated. A standard approach is the two step GMM procedure \citep{Hansen1982,Hayashi2000}.

\begin{enumerate}
    \item \textbf{Step 1: Initial estimation.}

    Compute the preliminary estimator $\hat{\theta}_n^{(1)}$ using
    $W_n=I_{2p}$, or another fixed positive definite matrix:
    \[
    \hat{\theta}_n^{(1)}
    =
    \arg\min_{\theta\in\Theta}
    \bar{\psi}_n(\theta)^\top
    \bar{\psi}_n(\theta).
    \]

    \item \textbf{Step 2: Efficient re-estimation.}

    Using $\hat{\theta}_n^{(1)}$, compute the sample covariance matrix of the moments:
    \[
    \widehat S_n
    =
    \frac{1}{n}
    \sum_{i=1}^n
    \psi\!\big(\mathcal U_i,\hat{\theta}_n^{(1)}\big)
    \psi\!\big(\mathcal U_i,\hat{\theta}_n^{(1)}\big)^\top.
    \]

    Set $W_n=\widehat S_n^{-1}$ and reestimate:
    \[
    \hat{\theta}_n^{(2)}
    =
    \arg\min_{\theta\in\Theta}
    \bar{\psi}_n(\theta)^\top
    \widehat S_n^{-1}
    \bar{\psi}_n(\theta).
    \]
\end{enumerate}

When the model is overidentified, the resulting two step estimator
$\hat{\theta}_n^{(2)}$ is asymptotically efficient within the class of GMM estimators based on the chosen moment conditions. In the just identified case considered in the present applications, the point estimator does not depend on the positive definite weighting matrix when the sample moment equations have an exact and unique solution \citep{Hansen1982}.

\subsubsection*{Computational implementation and choice of weighting matrix}
\label{computational_implementation}

The proposed ZAM regression models were estimated in R using the
\texttt{gmm} package \citep{Chausse2010}. The implementation uses the
two step GMM procedure with an identity weighting matrix in the first
step and an estimated optimal weighting matrix in the second step.

\paragraph{First step estimation.}

In the first step, the weighting matrix in \eqref{eq:estimator} is set equal to the identity matrix: \(
W_n^{(1)}=I_s.
\) This produces the preliminary estimator
\[
\hat\theta_n^{(1)}
=
\arg\min_{\theta\in\Theta}
\bar\psi_n(\theta)^\top
I_s
\bar\psi_n(\theta).
\]

The identity matrix is fixed and positive definite, and it does not require a prior estimate of the covariance matrix of the moment conditions. It therefore provides a simple starting point for the two step procedure \citep{Hansen1982,Hayashi2000}. The identity matrix does not imply that all moment conditions have the same variability. Its purpose is only to obtain a consistent preliminary estimate that can be used to estimate the covariance structure of the moments.

\paragraph{Second step estimation.}

The centered sample covariance matrix of the moment conditions is evaluated at
$\hat\theta_n^{(1)}$ as follows:
\[
\widehat S_n^{(1)}
=
\frac{1}{n}
\sum_{i=1}^n
\left\{
\psi(\mathcal U_i,\hat\theta_n^{(1)})
-
\bar\psi_n(\hat\theta_n^{(1)})
\right\}
\left\{
\psi(\mathcal U_i,\hat\theta_n^{(1)})
-
\bar\psi_n(\hat\theta_n^{(1)})
\right\}^{\top}.
\]

The second step weighting matrix is then set equal to
\[
W_n^{(2)}
=
\left(\widehat S_n^{(1)}\right)^{-1}.
\]

The GMM objective is minimized again using $W_n^{(2)}$:
\[
\hat\theta_n^{(2)}
=
\arg\min_{\theta\in\Theta}
\bar\psi_n(\theta)^\top
W_n^{(2)}
\bar\psi_n(\theta).
\]

The inverse covariance weighting matrix gives less weight to moment conditions with greater variability and accounts for correlations among the moment conditions. When the model is overidentified, so that $s>d$, this choice produces the asymptotically efficient GMM estimator within the class of estimators based on the selected moment conditions \citep{Hansen1982,Hayashi2000}. This result also follows from Theorem~\ref{thm:AN} by taking
\(
W_0=S_0^{-1}.
\)

\paragraph{Just identified case.}

In the applications considered in this paper, the number of moment conditions equals the number of parameters. This is the just identified case \citep{Hansen1982}. The constant dispersion model has $p+1$ moment conditions for $p+1$ parameters, while the varying dispersion model has $2p$ moment conditions for $2p$ parameters.

When the sample moment equations have an exact solution, the GMM estimator satisfies
\[
\bar\psi_n(\hat\theta_n)=0.
\]

It follows that
\(
Q_n(\hat\theta_n)
=
\bar\psi_n(\hat\theta_n)^\top
W_n
\bar\psi_n(\hat\theta_n)
=
0
\)
for every positive definite weighting matrix $W_n$. Thus, provided the solution is unique, the resulting point estimator does not depend on the choice of the weighting matrix. In numerical computation, the sample moments and the minimized objective may be close to zero rather than exactly zero because of the optimizer tolerance.

\paragraph{Starting values.}

Starting values for the mean parameters are obtained from a Poisson regression with the same mean design matrix. For the dispersion parameters, an initial overall dispersion estimate is used for the intercept, and the dispersion slopes are initialized at zero. The ZA parameter $\delta$ is not estimated directly in the regression model, since dispersion is modeled through $\phi$ and the regression parameter $\gamma$.

\paragraph{Optimization and predictor scaling.}

The GMM objective is minimized using the BFGS optimization method. Selected continuous predictors are centered and scaled before estimation to improve numerical stability \citep{NocedalWright2006}.

\subsubsection*{Plug-in covariance matrix and tests}

Let
\(
\textcolor{black}{\widehat S}
=
\frac{1}{n}\sum_{i=1}^n
\textcolor{black}{\psi}(\mathcal U_i,\hat\theta_n)\,
\textcolor{black}{\psi}(\mathcal U_i,\hat\theta_n)^\top
\)
and
\[
\hat G
=
-\begin{bmatrix}
n^{-1}\sum_{i=1}^n X_i\,g'(X_i^\top\hat\beta)\,X_i^\top & 0\\[2mm]
n^{-1}\sum_{i=1}^n X_i\,2\,g(X_i^\top\hat\beta)g'(X_i^\top\hat\beta)h(X_i^\top\hat\gamma)\,X_i^\top
&
n^{-1}\sum_{i=1}^n X_i\,g(X_i^\top\hat\beta)^2h'(X_i^\top\hat\gamma)\,X_i^\top
\end{bmatrix}.
\]
For a general weighting matrix \(\textcolor{black}{\hat W}\), define
\begin{equation*}
\textcolor{black}{\widehat V}
=
(\hat G^\top\textcolor{black}{\hat W}\hat G)^{-1}
\hat G^\top\textcolor{black}{\hat W}\,\textcolor{black}{\widehat S}\,\textcolor{black}{\hat W}\hat G
(\hat G^\top\textcolor{black}{\hat W}\hat G)^{-1}.
\end{equation*}
Then \(\textcolor{black}{\widehat V}\) estimates the asymptotic covariance matrix of \(\sqrt n(\hat\theta_n-\theta_0)\), and therefore
\begin{equation*}
\widehat{\mathrm{Var}}(\hat\theta_n)
=
\textcolor{black}{\frac{1}{n}\widehat V}.
\end{equation*}
For the optimal two-step GMM estimator, \(\textcolor{black}{\hat W=\widehat S^{-1}}\), and
\(
\textcolor{black}{\widehat V}
=
(\hat G^\top\textcolor{black}{\widehat S^{-1}}\hat G)^{-1}.
\)

\paragraph{Consistency of the plug-in covariance estimator.}
\begin{lemma}\label{lem:Shat_cons}
Under Assumptions~\ref{ass:A1}--\ref{ass:A5},
\(
\widehat S \xrightarrow{p} S_0.
\)
\end{lemma}

\begin{lemma}\label{lem:What_cons}
If \(S_0\succ0\), then
\(
\widehat S^{-1}\xrightarrow{p}S_0^{-1}.
\)
\end{lemma}

\begin{proposition}\label{prop:Vhat_cons}
Let
\(
\widehat V
=
(\hat G^\top \hat W \hat G)^{-1}
\hat G^\top \hat W\,\widehat S\,\hat W \hat G
(\hat G^\top \hat W \hat G)^{-1},
\)
where \(\hat G\) is the sample Jacobian defined above. If \(\hat W\xrightarrow{p}W_0\), then
\[
\widehat V\xrightarrow{p}
(G^\top W_0G)^{-1}G^\top W_0S_0W_0G(G^\top W_0G)^{-1}.
\]
For the optimal two-step choice \(\hat W=\widehat S^{-1}\),
\[
\widehat V\xrightarrow{p}(G^\top S_0^{-1}G)^{-1}.
\]
Therefore,
\[
\widehat{\mathrm{Var}}(\hat\theta_n)=\frac{1}{n}\widehat V.
\]
\end{proposition}

\begin{proposition}[J test]\label{prop:J}
With \(X_i\) used as the instrument vector in both moment blocks, the model is just identified. The number of moment conditions is \(2p\), which matches the number of parameters in \(\theta\). Hence,
\[
J_n
=
n\,\bar{\textcolor{black}{\psi}}_n(\hat\theta_n)^\top\textcolor{black}{\hat W}\,\bar{\textcolor{black}{\psi}}_n(\hat\theta_n)
\]
is asymptotically degenerate at zero and does not provide a nontrivial overidentification test.
\end{proposition}

\begin{proposition}[Wald tests]\label{prop:Wald}
For testing the linear restrictions
\[
H_0:R\theta_0=r,
\]
where \(R\) has full row rank, the Wald statistic is
\[
\mathcal W_n
=
(R\hat\theta_n-r)^\top
\left(
\frac{1}{n}R\widehat V R^\top
\right)^{-1}
(R\hat\theta_n-r).
\]
Under \(H_0\),
\[
\mathcal W_n
\xrightarrow{d}
\chi^2_{\mathrm{rank}(R)}.
\]
\end{proposition}

\section{Numerical Illustrations based on Simulated Data}
\label{sec:sim}

We assess the finite sample performance of the proposed GMM estimator under controlled data generating processes that cover both underdispersion and overdispersion. Each simulated data set contains $n=1000$ observations, and the random seed is fixed at $2025$ for reproducibility. An analytical gradient is used within the gradient based optimization procedure to improve computational efficiency \citep{NocedalWright2006}. Across all designs, we focus on the accuracy of point estimation and the reliability of inference. We report Monte Carlo (MC) means, standard deviations, average GMM standard errors, bias, and empirical coverage of nominal $95\%$ Wald intervals. We also summarize the distribution of estimation bias across replications using bias boxplots for the mean block $(\beta_0,\beta_1)$ and the dispersion block $(\gamma_0,\gamma_1)$.

\subsection{Under-dispersion}
\label{subsec:sim_under}

Table~\ref{tab:mc_under_only} indicates that the estimator recovers the true parameters for both under-dispersed baselines. The MC means are close to the truth and the biases are small for all parameters. The empirical coverage probabilities range from $0.92$ to $0.95$. Coverage is close to $0.95$ for several parameters, although it is somewhat below the nominal level for the ZABIN slopes and for some ZACMP dispersion parameters.


\begingroup
\renewcommand{\BigTblFont}{\footnotesize}
\renewcommand{\BigTblScale}{0.95}
\renewcommand{\BigTblWidth}{1.00\linewidth}

\begin{table}
\centering
{\BigTblFont
\renewcommand{\arraystretch}{1.15}
\setlength{\tabcolsep}{8.5pt}
\caption{Monte Carlo (MC) summary for under dispersed zero altered (ZA) models. We report results for the ZA Binomial model with $N$ trials (ZABIN($N$)) and the ZA COM Poisson model (ZACMP). For each parameter, we show the true value, the MC mean estimate (Mean Est.), the Monte Carlo standard deviation of the estimates (MC SD), the mean of the estimated GMM standard errors (Avg. GMM SE), the MC bias (Bias = Mean Est. minus True Value), and the empirical coverage probability of nominal 95 percent Wald intervals (Coverage probability). The last row reports elapsed runtime in seconds and the total number of objective function evaluations.}\label{tab:mc_under_only}

\begin{adjustbox}{max width=\BigTblWidth,center}
\scalebox{\BigTblScale}{%
\begin{tabular}{
|l|
c|
S[table-format=-1.4]|
S[table-format=-1.4]|
S[table-format=-1.4]|
S[table-format=-1.4]|
S[table-format=-1.4]|
S[table-format=-1.4]|
}
\hline
\textbf{Distribution}
& \textbf{Parameter}
& \multicolumn{1}{c|}{\shortstack{\textbf{True}\\\textbf{Value}}}
& \multicolumn{1}{c|}{\shortstack{\textbf{Mean}\\\textbf{Est.}}}
& \multicolumn{1}{c|}{\shortstack{\textbf{MC}\\\textbf{SD}}}
& \multicolumn{1}{c|}{\shortstack{\textbf{Avg. GMM}\\\textbf{SE}}}
& \multicolumn{1}{c|}{\textbf{Bias}}
& \multicolumn{1}{c|}{\shortstack{\textbf{Coverage}\\\textbf{probability}}}
\\
\hline

\multirow{4}{*}{\textbf{ZABIN($N$)}}
& $\beta_0$  &  0.4040 &  0.4039 & 0.0205 & 0.0221 & -0.0016 & 0.9500 \\
& $\beta_1$  &  0.1010 &  0.1014 & 0.0409 & 0.0381 &  0.0014 & 0.9167 \\
& $\gamma_0$ & -0.2020 & -0.2024 & 0.0275 & 0.0294 & -0.0024 & 0.9533 \\
& $\gamma_1$ & -0.1020 & -0.1017 & 0.0520 & 0.0505 & -0.0017 & 0.9200 \\
\hline

\multirow{4}{*}{\textbf{ZACMP}}
& $\beta_0$  & 1.2000 &  1.1993 & 0.0180 & 0.0177 & -0.0007 & 0.9433 \\
& $\beta_1$  & 0.5000 &  0.5020 & 0.0304 & 0.0301 &  0.0020 & 0.9467 \\
& $\gamma_0$ & 0.0000 & -0.0016 & 0.0154 & 0.0150 & -0.0016 & 0.9300 \\
& $\gamma_1$ & 0.0000 & -0.0009 & 0.0295 & 0.0283 & -0.0009 & 0.9233 \\
\hline

\multicolumn{8}{|c|}{
\begin{minipage}[c]{0.96\linewidth}
\vspace{1.5mm}
\footnotesize
\textbf{Bias tests:}
None of the estimated biases was statistically significant at the
5\% level; p-values ranged from 0.0671 to 0.5875.
\vspace{1.5mm}
\end{minipage}
}
\\
\hline

\multicolumn{8}{|c|}{
\begin{minipage}[c]{0.96\linewidth}
\vspace{1.5mm}
\footnotesize
\textbf{Runtime (elapsed s; function evals):}
ZABIN($N$): ---; 27.
\qquad
ZACMP: 58.0400; 48.
\vspace{1.5mm}
\end{minipage}
}
\\
\hline

\end{tabular}%
}
\end{adjustbox}

} 
\end{table}
\endgroup


\begin{figure}
  \centering

  \begin{minipage}[b]{0.48\linewidth}
    \centering
    \includegraphics[width=\linewidth]{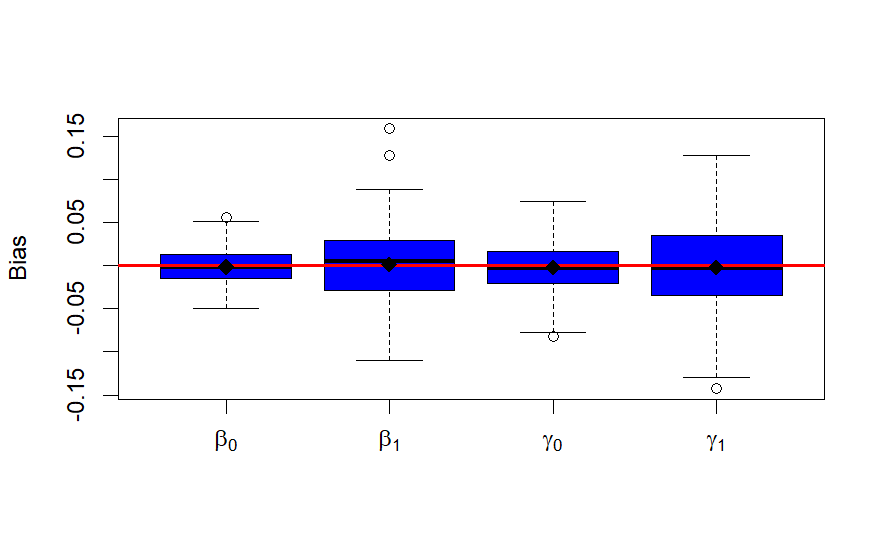}
    \par\smallskip
    (a) ZA Binomial
  \end{minipage}
  \hfill
  \begin{minipage}[b]{0.48\linewidth}
    \centering
    \includegraphics[width=\linewidth]{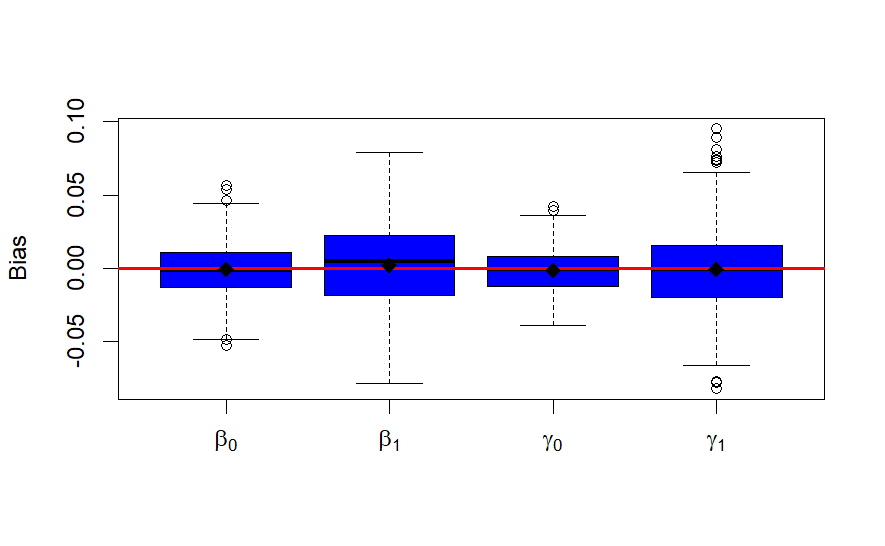}
    \par\smallskip
    (b) ZA CMP
  \end{minipage}

  \caption{Bias boxplots for zero-altered (ZA) models, under-dispersed cases: (a) ZA-Binomial with $N$ trials, and (b) ZA-COM-Poisson (ZACMP).}
  \label{fig:bias_under}
\end{figure}

We validated our findings in Table~\ref{tab:mc_under_only} with Figure~\ref{fig:bias_under} which indicate that the bias boxplots are centered near zero for all parameters in both models. The interquartile ranges are small and the medians lie close to the horizontal zero line, which supports the conclusion that the estimator is asymptotically unbiased in the under-dispersed cases.

\subsection{Over-dispersion}
\label{subsec:sim_over}

Table~\ref{tab:mc_over_only} shows strong finite sample performance across the over-dispersed designs. For the ZIP, ZA-NB2, and ZA-PIG models, the Monte Carlo means are close to the true parameters and the biases are small for both the mean block $(\beta_0,\beta_1)$ and the dispersion block $(\gamma_0,\gamma_1)$. The standard deviations are generally close to the average GMM standard errors, which supports the accuracy of the reported standard errors. The main exception is $\gamma_1$ for ZAPIG, where the MC SD is $0.0650$ and the average GMM SE is $0.0517$. Coverage is near the nominal $95\%$ level for most parameters, ranging from $0.8950$ to $0.9750$, with the lowest coverage occurring for ZA-PIG $\gamma_1$.


\begingroup
\renewcommand{\BigTblFont}{\footnotesize}
\renewcommand{\BigTblScale}{0.92}
\renewcommand{\BigTblWidth}{1.00\linewidth}

\begin{table}
\centering

{\BigTblFont
\renewcommand{\arraystretch}{1.15}
\setlength{\tabcolsep}{8.2pt}
\caption{Monte Carlo (MC) summary for over dispersed zero altered (ZA) models. We report results for the zero inflated Poisson model (ZIP), the ZA negative binomial model with NB2 variance function (ZANB2), and the ZA Poisson inverse Gaussian model (ZAPIG). For each parameter in each model, we report the true value, the MC mean estimate (Mean Est.), the Monte Carlo standard deviation of the estimates (MC SD), the mean of the estimated GMM standard errors (Avg. GMM SE), the MC bias (Bias = Mean Est. minus True Value), and the empirical coverage probability of nominal 95 percent Wald intervals (Coverage probability).}

\label{tab:mc_over_only}

\begin{adjustbox}{max width=\BigTblWidth,center}
\scalebox{\BigTblScale}{%
\begin{tabular}{
|l|
c|
S[table-format=-1.4]|
S[table-format=-1.4]|
S[table-format=-1.4]|
S[table-format=-1.4]|
S[table-format=-1.4]|
S[table-format=-1.4]|
}
\hline
\textbf{Distribution}
& \textbf{Parameter}
& \multicolumn{1}{c|}{\shortstack{\textbf{True}\\\textbf{Value}}}
& \multicolumn{1}{c|}{\shortstack{\textbf{Mean}\\\textbf{Est.}}}
& \multicolumn{1}{c|}{\shortstack{\textbf{MC}\\\textbf{SD}}}
& \multicolumn{1}{c|}{\shortstack{\textbf{Avg. GMM}\\\textbf{SE}}}
& \multicolumn{1}{c|}{\textbf{Bias}}
& \multicolumn{1}{c|}{\shortstack{\textbf{Coverage}\\\textbf{probability}}}
\\
\hline

\multirow{4}{*}{\textbf{ZIP}}
& $\beta_0$  & 1.0980 & 1.0983 & 0.0232 & 0.0239 &  0.0003 & 0.9580 \\
& $\beta_1$  & 0.3980 & 0.3977 & 0.0462 & 0.0435 & -0.0003 & 0.9400 \\
& $\gamma_0$ & 0.2000 & 0.1984 & 0.0219 & 0.0215 & -0.0016 & 0.9420 \\
& $\gamma_1$ & 0.2000 & 0.1994 & 0.0426 & 0.0399 & -0.0006 & 0.9340 \\
\hline

\multirow{4}{*}{\textbf{ZANB2}}
& $\beta_0$  & 1.2000 & 1.2005 & 0.0187 & 0.0177 &  0.0005 & 0.9240 \\
& $\beta_1$  & 0.5000 & 0.4998 & 0.0294 & 0.0302 & -0.0002 & 0.9600 \\
& $\gamma_0$ & 0.0000 & -0.0007 & 0.0151 & 0.0149 & -0.0007 & 0.9480 \\
& $\gamma_1$ & 0.0000 & -0.0005 & 0.0287 & 0.0283 & -0.0005 & 0.9420 \\
\hline

\multirow{4}{*}{\textbf{ZAPIG}}
& $\beta_0$  & 0.6000 & 0.6006 & 0.0255 & 0.0274 &  0.0006 & 0.9750 \\
& $\beta_1$  & 0.4000 & 0.3979 & 0.0298 & 0.0290 & -0.0021 & 0.9400 \\
& $\gamma_0$ & 0.1500 & 0.1445 & 0.0470 & 0.0445 & -0.0055 & 0.9500 \\
& $\gamma_1$ & 0.0500 & 0.0479 & 0.0650 & 0.0517 & -0.0021 & 0.8950 \\
\hline

\multicolumn{8}{|c|}{
\begin{minipage}[c]{0.96\linewidth}
\vspace{1.5mm}
\footnotesize
\textbf{Bias tests:}
None of the estimated biases was statistically significant at the
5\% level; p-values ranged from 0.0978 to 0.8536.
\vspace{1.5mm}
\end{minipage}
}
\\
\hline

\multicolumn{8}{|c|}{
\begin{minipage}[c]{0.96\linewidth}
\vspace{1.5mm}
\footnotesize
\textbf{Runtime (elapsed s; function evals):}
ZIP: 7.7100; 48.
\qquad
ZANB2: 7.9200; 42.
\qquad
ZAPIG: 4.8900; 42.
\vspace{1.5mm}
\end{minipage}
}
\\
\hline

\end{tabular}%
}
\end{adjustbox}

} 
\end{table}
\endgroup


\begin{figure}
  \centering

  \begin{minipage}[b]{0.32\linewidth}
    \centering
    \includegraphics[width=\linewidth]{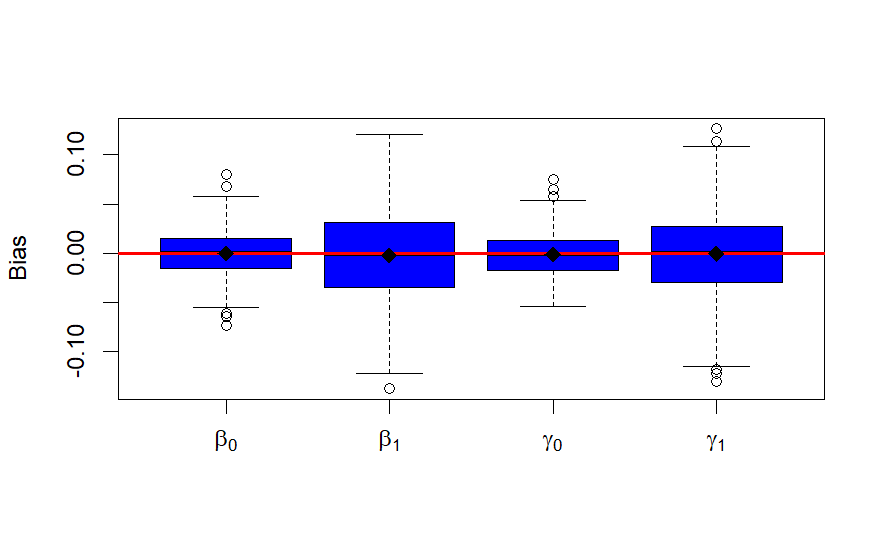}
    \par\smallskip
    (a) ZIP
  \end{minipage}
  \hfill
  \begin{minipage}[b]{0.32\linewidth}
    \centering
    \includegraphics[width=\linewidth]{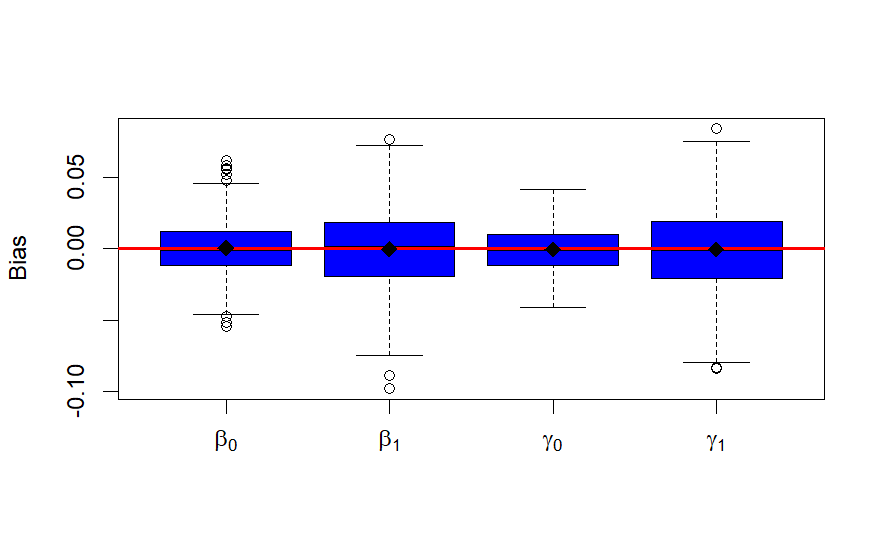}
    \par\smallskip
    (b) ZANB2
  \end{minipage}
  \hfill
  \begin{minipage}[b]{0.32\linewidth}
    \centering
    \includegraphics[width=\linewidth]{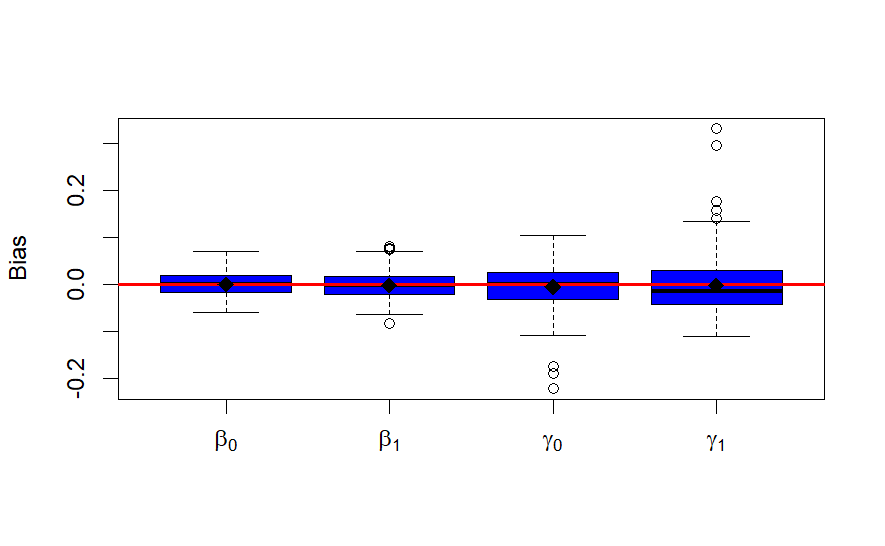}
    \par\smallskip
    (c) ZAPIG
  \end{minipage}

  \caption{Boxplot of biases for zero-altered (ZA) models, over-dispersed cases: zero-inflated Poisson (ZIP), ZA-negative binomial (ZANB2), and ZA-Poisson-inverse Gaussian (ZAPIG).}
  \label{fig:bias_over}
\end{figure}

We can also observe from Figure~\ref{fig:bias_over} that the bias boxplots are centered near zero in all three panels, with medians close to the horizontal zero line. The spread is small for ZIP and ZANB2, which indicates stable estimation. Based on these simulations, the proposed moments provide good finite sample performance under the designs considered, although the lower coverage for ZAPIG $\gamma_1$ suggests some caution for this parameter.

\section{Applications to Real Data Sets}
\label{RealData}
We fit our moment-based count model to three real datasets. The first example is the Attendance Records dataset (contained in the {\tt mixpoissonreg} package) --  an attendance dataset that reports the number of absences of 314 urban high school juniors. For each student, the dataset reports the number of days absent along with the student's gender, program track (General, Academic, Vocational), and standardized math score \citep{UCLA_NBR_Attendance}. The second dataset is the Takeover Bids dataset (contained in the {\tt CountsEPPM} package) which tracks the number of takeover bids from United States firms that faced tender offers between 1978 and 1985. For each company, the dataset records the number of bids and includes the opening premium, whether managers used legal or real asset restrictions, whether a white knight appeared, how much stock institutions held as well as firm size \citep{CameronJohansson1997,Huang2017CMPmu}.
The third data example stems from a Cotton bolls agricultural experiment where the study counts the number of cotton bolls observed on plants in one of five plant stages (vegetative, bud, blossom, fig, and boll). For each record,  the study also reports the defoliation level \citep{Fung_mpcmp, Zeviani2014GammaCount}.

\subsection{Exploratory data analysis}
\label{sec:eda}

Before fitting the regression models, we first examine dispersion diagnostics associated with the respective datasets. For each dataset, we form covariate groups and compute the within-group sample mean and sample variance of the count response. We then summarize dispersion using the log ratio,
\[
D=\log\left\{\frac{\widehat{\Var}(Y\mid \text{group})}{\widehat{\E}(Y\mid \text{group})}\right\}.
\]
Under Poisson behavior, we expect $\Var(Y\mid \text{group}) \approx \E(Y\mid \text{group})$, so $D$ should be close to zero. Positive values of $D$ indicate over-dispersion and negative values indicate under-dispersion. The horizontal dashed line at zero serves as the reference line for equi-dispersion. 

We apply this diagnostic to the three data sets. For the Attendance data, we group observations by program type and gender/sex, and we compute $D$ across quantile bins of the math score. This shows how dispersion changes with academic track and performance level. For the Takeover Bids dataset, we group by the presence of a white knight and by real restriction status, and we compute $D$ across quartiles of firm size. For the Cotton Bolls data, we group by growth stage and compute $D$ across quantile bins of the defoliation measure. Across all three examples, the diagnostic values vary across groups and, in some cases, across zero, which suggests that dispersion can depend on covariates and may not be well captured by a single Poisson or NB2 specification. These patterns motivate our regression framework that models both the mean and the dispersion as functions of covariates and estimate parameters using moment conditions rather than full likelihood specification.

\begin{figure}
    \centering
    \includegraphics[width=1\linewidth]{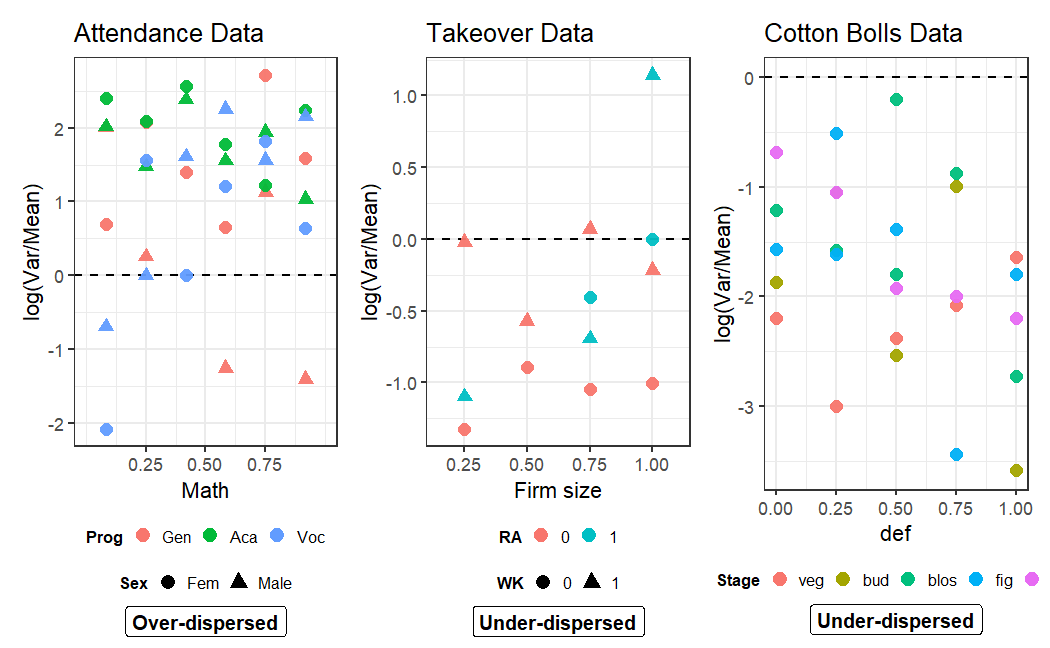}
    \caption{Dispersion diagnostic $\log\{\widehat{\Var}(Y)/\widehat{\E}(Y)\}$ across covariate groups for the Attendance, Takeover Bids, and Cotton Bolls data, where ``def'' denotes defoliation level quantile respectively. The dashed line at 0 marks equi-dispersion .}
    \label{EDA_For_ATC}
\end{figure}

From Figure~\ref{EDA_For_ATC},
each point is the value of $\log\{\widehat{\Var}(Y)/\widehat{\E}(Y)\}$ computed within one covariate group (the group is defined by the legend categories and the bin on the x-axis).
In the Attendance panel, most points are above the zero line across math levels, so the counts are mainly over-dispersed.  In the Takeover Bids panel, many groups fall below zero for small and medium firm size, which indicates under-dispersion in several groups. At the largest firm size, however, a few groups move toward zero and one group is above zero, suggesting that dispersion varies with firm size and the indicators. In the Cotton Bolls panel, almost all points are below zero across defoliation level (def), which implies strong under-dispersion for most stages. Differences across stages also suggest that dispersion is not constant and should be allowed to depend on covariates.

\subsection{Analyses and Results for the Respective Examples}
In our model, we link the mean count to covariates via a log link, and associate the dispersion index through the map \(1+\phi=\exp(z^\top\gamma)\). This allows \(\phi\) to be less than zero for under-dispersed data or greater than zero for over-dispersed data. We use the first two moments, and compare models where one considers a constant \(\phi\), and another where \(\phi\) changes with covariates. This design preserves the mean effects from the generalized linear models; this also lets the data show where the spread is constant and where the spread changes with observed features. Prior work on these datasets often use models with a single dispersion number. For our analysis, we check that assumption and report tests and diagnostics.

\subsubsection{Case Study I: Attendance Data}
\label{case1}

\begingroup
\renewcommand{\BigTblFont}{\footnotesize}
\renewcommand{\BigTblScale}{0.90}
\renewcommand{\BigTblWidth}{1.00\linewidth}
\renewcommand{\DistW}{16mm}

\begin{table}
\caption{Attendance: compact comparison of constant vs varying $\phi$ with named parameters (zamreg, GMM).}
\label{tab:attendance_compact_named}
\centering
{\BigTblFont
\renewcommand{\arraystretch}{1.12}
\setlength{\tabcolsep}{3pt}
\begin{adjustbox}{max width=\BigTblWidth,center}
\scalebox{\BigTblScale}{%
\begin{tabular}{|l|*{8}{P{\DistW}|}}
\hline
\multicolumn{1}{|c|}{\textbf{Param}} &
\multicolumn{4}{c|}{\textbf{Constant $\phi$}} &
\multicolumn{4}{c|}{\textbf{Varying $\phi$}} \\
\cline{2-9}
 & \textbf{Est} & \textbf{SE} & \textbf{t} & \textbf{p} & \textbf{Est} & \textbf{SE} & \textbf{t} & \textbf{p} \\
\hline
$\beta_0$ (Intercept)             & 2.7631 & 0.1557 & 17.7425 & 0.0000 & 2.7640 & 0.1556 & 17.7634 & 0.0000 \\
$\beta_{\text{male}}$             & -0.2412 & 0.1197 & -2.0155 & 0.0438 & -0.2424 & 0.1197 & -2.0244 & 0.0429 \\
$\beta_{\text{progAcademic}}$     & -0.4322 & 0.1389 & -3.1117 & 0.0019 & -0.4343 & 0.1387 & -3.1308 & 0.0017 \\
$\beta_{\text{progVocational}}$   & -1.2854 & 0.1799 & -7.1430 & 0.0000 & -1.2910 & 0.1801 & -7.1690 & 0.0000 \\
$\beta_{\text{math}}$             & -0.0069 & 0.0024 & -2.9337 & 0.0033 & -0.0069 & 0.0024 & -2.9173 & 0.0035 \\
\hline
$\gamma_0$ (Intercept)            & 0.5739 & 0.0465 & 12.3548 & 0.0000 & 0.3348 & 0.1053 &  3.1795 & 0.0015 \\
$\gamma_{\text{male}}$            &   --   &   --   &   --    &   --   & -0.0289 & 0.1049 & -0.2753 & 0.7831 \\
$\gamma_{\text{progAcademic}}$    &   --   &   --   &   --    &   --   & 0.3285 & 0.0942 &  3.4864 & 0.0005 \\
$\gamma_{\text{progVocational}}$  &   --   &   --   &   --    &   --   & 0.5947 & 0.1342 &  4.4325 & 0.0000 \\
$\gamma_{\text{math}}$            &   --   &   --   &   --    &   --   & 0.0001 & 0.0020 &  0.0610 & 0.9513 \\
\hline
\multicolumn{9}{|p{\dimexpr\linewidth-2\fboxsep\relax}|}{
\textbf{Notes:} Constant $\phi=0.7752$. \;
Varying $\phi(x)$: mean$=1.0569$, min$=0.3581$, max$=1.5635$. \;
Wald (dispersion slopes): $\chi^2(4)=24.6712$, $p=0.0001$. \;
\textbf{Runtime (proc.time):} constant (user 0.01, system 0.00, elapsed 0.01), varying (user 0.02, system 0.00, elapsed 0.03).
}\\
\hline
\end{tabular}%
}
\end{adjustbox}
} 
\end{table}
\endgroup


From Table~\ref{tab:attendance_compact_named}, the baseline student is a girl in the General track with a math score of zero. Her expected absences are about $e^{2.7631}\approx15.85$, or approximately $16$ days. Boys have a lower absentee rate at approximately $e^{-0.2412}\approx0.79$ times the girls' rate (i.e. approximately 21\% fewer days). The Academic track is associated with a rate of $e^{-0.4322}\approx0.65$ of the General track rate, which is about 35\% fewer days, and this effect is statistically significant in both models ($p=0.0019$ for constant $\phi$ and $p=0.0017$ for varying $\phi$). The Vocational track has the largest decrease, with $e^{-1.2854}\approx0.28$ of the General track rate, so about 72\% fewer days, and this effect is also strongly significant in both models ($p<0.0001$). Also, a 10 point increase in math score is associated with $e^{-0.0069\times10}\approx0.93$, about seven percent fewer days, which is significant ($p=0.0033$ for constant $\phi$ and $p=0.0035$ for varying $\phi$). These estimates are almost unchanged when dispersion is allowed to vary, implying that the mean effects are stable.

The spread tells us how noisy the counts are around the mean. If we keep one spread for everyone, we estimate $\phi=0.7752$. This means that the variance is $\mu+0.7752\mu^2$, which is more variable than a Poisson. When we let the spread depend on sex/gender, program and math, the baseline dispersion intercept is $\gamma_0=0.3348$, which corresponds to $\phi=e^{0.3348}-1\approx0.3977$ for a girl in the General track with a math score of zero. The slopes show clear differences by program: progAcademic has $\gamma_{\text{progAcademic}}=0.3285$ ($p=0.0005$) and progVocational has $\gamma_{\text{progVocational}}=0.5947$ ($p<0.0001$), while male and math are not important for the spread (both are not significant). The joint test of the dispersion slopes gives $\chi^2(4)=24.6712$ with $p=0.0001$, so the constant spread assumption is not supported here. In other words, the data are over-dispersed overall, and the amount of extra spread changes across programs.

This has direct meaning for schools and districts. The mean results say program design is a strong lever: Academic and Vocational tracks are linked to fewer absences relative to the General track. The dispersion results add a second message: even after accounting for mean differences, the Academic and Vocational groups have more uneven attendance patterns (more variability around the mean). Practically, this means stakeholders should not only look at average attendance, but also stability. A district can use this to target support: the mean effects suggest which programs are associated with better attendance on average, while the varying dispersion suggests where attendance is less predictable and may need extra monitoring, early warning systems, transportation or scheduling support, and flexible interventions. For an education board or government agency, this helps planning and budgeting, because higher dispersion implies more uncertainty in expected absences across students and classrooms, so resources for attendance teams and support services can be allocated more efficiently to the programs where variability is higher. Many applied examples fit a single dispersion number for these data, for instance the standard negative binomial treatment used in the UCLA guide \citep{UCLA_NBR_Attendance}. Our results show that a single number can miss real differences in spread across programs, even though the mean effects are stable.

\begin{figure}
  \centering
  \includegraphics[width=0.75\linewidth]{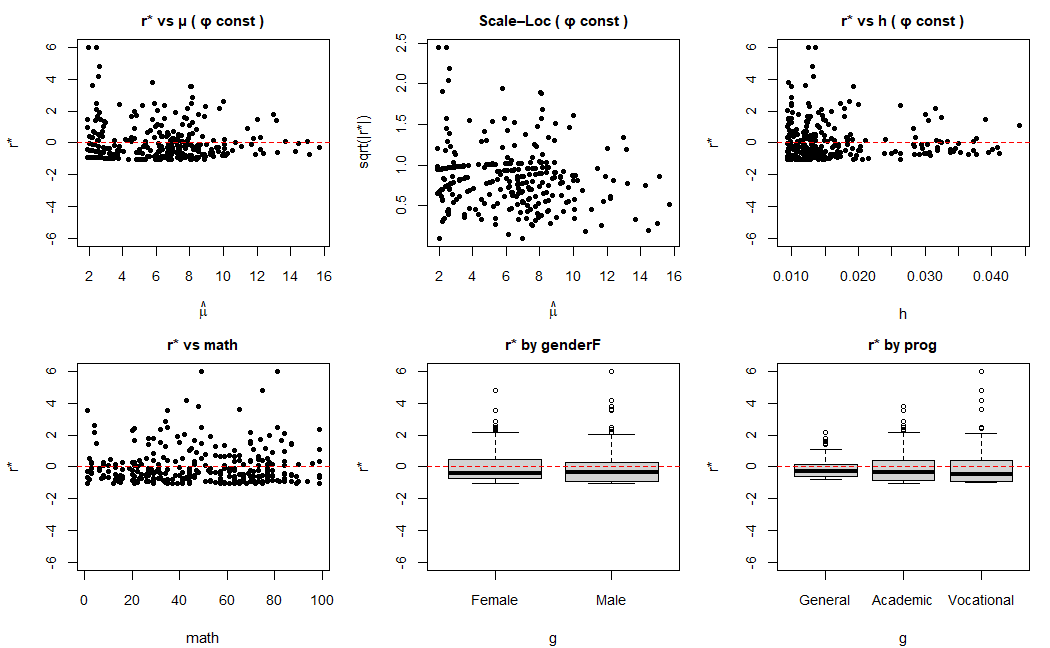}
  \caption{Attendance diagnostics under constant $\phi$.}
  \label{fig:att-diag-1}
\end{figure}

\begin{figure}
  \centering
  \includegraphics[width=0.75\linewidth]{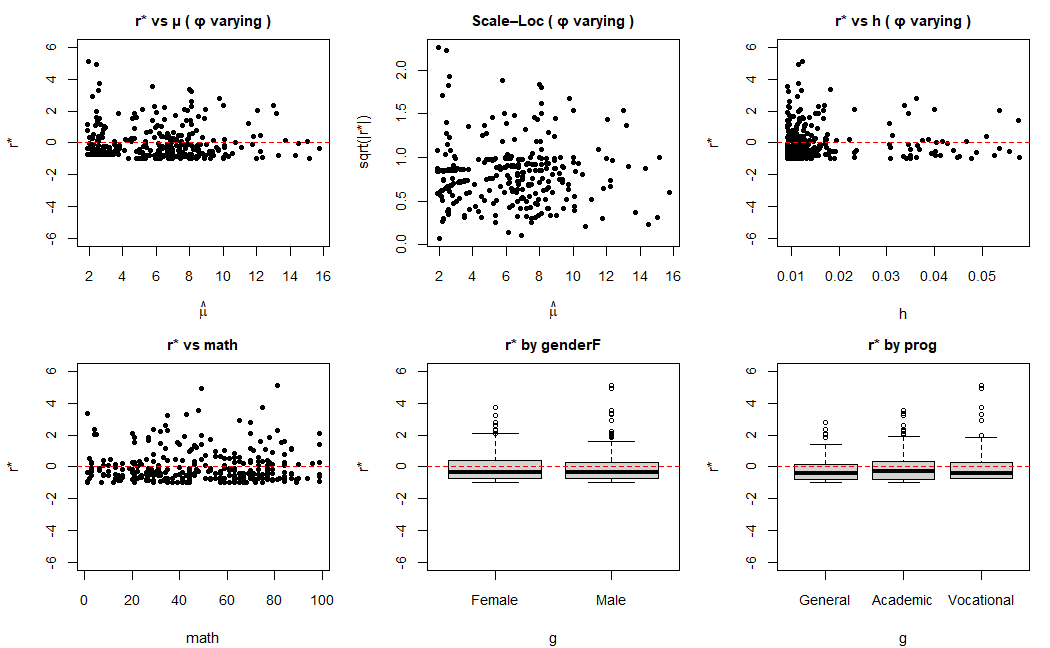}
  \caption{Attendance diagnostics under varying $\phi$.}
  \label{fig:att-diag-2}
\end{figure}

Figures \ref{fig:att-diag-1} and \ref{fig:att-diag-2} plot standardized Pearson residuals \(r^{*}\) for the constant \(\phi\) fit (left) and for the varying \(\phi\) fit (right). In both columns, the cloud of \(r^{*}\) values is centered close to zero and almost all points lie between \(-3\) and \(+3\). The plot of \(r^{*}\) against the fitted mean \(\hat\mu\) shows no clear trend and no strong funnel shape. Again we can see that the scale location panels are fairly flat, which indicates that both variance models describe much of the residual spread. The plot of \(r^{*}\) against leverage \(h\) also shows no notable pattern. The plot of \(r^{*}\) against math is nearly flat and the boxplots by gender and by program have broadly similar centers and spreads. Although the two sets of panels look similar for constant \(\phi\) and varying \(\phi\), the formal test in Table~\ref{tab:attendance_compact_named} shows that the dispersion slopes are jointly important \((\chi^{2}(4)=24.6712,\ p=0.0001)\). Thus, the visual differences are subtle, but the constant spread assumption is not supported, mainly because dispersion differs across programs.

\subsubsection{Case Study II: Takeover Bids Data}
\label{case2}

\begingroup
\renewcommand{\BigTblFont}{\footnotesize}
\renewcommand{\BigTblScale}{0.90}
\renewcommand{\BigTblWidth}{1.00\linewidth}
\renewcommand{\DistW}{16mm}

\begin{table}
\caption{Takeover Bids: compact comparison of constant vs varying $\phi$ with named parameters.}
\label{tab:takeover_compact_named}
\centering
{\BigTblFont
\renewcommand{\arraystretch}{1.12}
\setlength{\tabcolsep}{3pt}
\begin{adjustbox}{max width=\BigTblWidth,center}
\scalebox{\BigTblScale}{%
\begin{tabular}{|l|*{8}{P{\DistW}|}}
\hline
\multicolumn{1}{|c|}{\textbf{Param}} &
\multicolumn{4}{c|}{\textbf{Constant $\phi$}} &
\multicolumn{4}{c|}{\textbf{Varying $\phi$}} \\
\cline{2-9}
 & \textbf{Est} & \textbf{SE} & \textbf{t} & \textbf{p} & \textbf{Est} & \textbf{SE} & \textbf{t} & \textbf{p} \\
\hline
$\beta_0$ (Intercept)            & 1.0610 & 0.4175 &  2.5412 & 0.0110 & 1.0593 & 0.4177 &  2.5360 & 0.0112 \\
$\beta_{\text{BIDPREM}}$         & -0.6575 & 0.3034 & -2.1670 & 0.0302 & -0.6581 & 0.3035 & -2.1684 & 0.0301 \\
$\beta_{\text{INSTHOLD}}$        & -0.3725 & 0.3190 & -1.1677 & 0.2429 & -0.3716 & 0.3192 & -1.1640 & 0.2444 \\
$\beta_{\text{LEGLREST}}$        & 0.2607  & 0.1245 &  2.0941 & 0.0362 & 0.2588  & 0.1246 &  2.0770 & 0.0378 \\
$\beta_{\text{REALREST}}$        & -0.1121 & 0.1946 & -0.5762 & 0.5645 & -0.1059 & 0.1939 & -0.5461 & 0.5850 \\
$\beta_{\text{FINREST}}$         & 0.0810  & 0.2543 &  0.3184 & 0.7502 & 0.0838  & 0.2540 &  0.3299 & 0.7415 \\
$\beta_{\text{REGULATN}}$        & -0.0514 & 0.1504 & -0.3416 & 0.7327 & -0.0462 & 0.1503 & -0.3074 & 0.7585 \\
$\beta_{\text{WHTKNGHT1}}$       & 0.5103  & 0.1067 &  4.7824 & 0.0000 & 0.5106  & 0.1069 &  4.7779 & 0.0000 \\
$\beta_{\text{logSIZEz}}$        & 0.1640  & 0.0798 &  2.0554 & 0.0398 & 0.1624  & 0.0798 &  2.0342 & 0.0419 \\
\hline
$\gamma_0$ (Intercept)           & -0.0416 & 0.1068 & -0.3891 & 0.6972 & -0.2551 & 0.7262 & -0.3513 & 0.7254 \\
$\gamma_{\text{BIDPREM}}$        &   --    &   --   &   --    &   --   & -0.6226 & 0.4969 & -1.2529 & 0.2102 \\
$\gamma_{\text{INSTHOLD}}$       &   --    &   --   &   --    &   --   & -0.0118 & 0.5476 & -0.0215 & 0.9828 \\
$\gamma_{\text{LEGLREST}}$       &   --    &   --   &   --    &   --   & 0.1455  & 0.2572 &  0.5656 & 0.5716 \\
$\gamma_{\text{REALREST}}$       &   --    &   --   &   --    &   --   & 0.7521  & 0.2472 &  3.0428 & 0.0023 \\
$\gamma_{\text{FINREST}}$        &   --    &   --   &   --    &   --   & 0.2442  & 0.1944 &  1.2559 & 0.2091 \\
$\gamma_{\text{REGULATN}}$       &   --    &   --   &   --    &   --   & 0.1003  & 0.2622 &  0.3824 & 0.7022 \\
$\gamma_{\text{WHTKNGHT1}}$      &   --    &   --   &   --    &   --   & 0.6183  & 0.2173 &  2.8447 & 0.0044 \\
$\gamma_{\text{logSIZEz}}$       &   --    &   --   &   --    &   --   & 0.0672  & 0.0554 &  1.2124 & 0.2254 \\
\hline
\multicolumn{9}{|p{\dimexpr\linewidth-2\fboxsep\relax}|}{
\textbf{Notes:} Constant $\phi=-0.0407$. \;
Varying $\phi(x)$: mean$=-0.2722$, min$=-0.7493$, max$=1.4081$. \;
Wald (dispersion slopes): $\chi^2(8)=46.1476$, $p=0.0000$. \;
\textbf{Runtime (proc.time):} constant (user 0.01, system 0.00, elapsed 0.00), varying (user 0.06, system 0.04, elapsed 0.11).
}\\
\hline
\end{tabular}%
}
\end{adjustbox}
} 
\end{table}
\endgroup


Earlier studies often used one common number for the spread of the counts. \cite{Huang2017CMPmu} fits a mean-parametrized COM-Poisson model with a single dispersion term and \cite{SaezCastillo2013HyperPoisson} use a hyper-Poisson model that also treats dispersion as constant. Our results in Table~\ref{tab:takeover_compact_named} meanwhile add a new component. The mean results agree with past work. When a white knight appears, the expected number of bids is about 67\% higher \(\exp(0.5103)\approx1.67\) with \(p<0.0001\). When the first offer has a higher premium, the expected number of bids is about 52\% as large \(\exp(-0.6575)\approx0.52\), or approximately 48\% lower, with \(p=0.0302\). A legal restriction flag is linked to about 30\% more bids \(\exp(0.2607)\approx1.30\) with \(p=0.0362\). A one unit increase in standardized log firm size is associated with about 18\% more bids \(\exp(0.1640)\approx1.18\) with \(p=0.0398\).

The dispersion component provides additional information beyond the mean model. If we force one common spread, we get a small fixed under-dispersion \(\phi=-0.0407\). When we let the spread change with the same covariates, the average becomes \(\phi=-0.2722\) and it ranges from \(-0.7493\) to \(1.4081\). The dispersion slopes matter as a group with \(\chi^2(8)=46.1476\) and \(p<0.0001\). Most important for policy, the real asset restriction has a strong and precise effect on spread. The dispersion link coefficient estimate for real asset restriction is \(0.7521\) (SE \(=0.2472\), \(t=3.0428\), \(p=0.0023\)); this multiplies \(1+\phi\) by about \(e^{0.7521}\approx2.12\), thus the real asset restrictions make outcomes much more variable. The white knight also raises spread with an estimated coefficient of \(\hat{\gamma}_{\text{WHTKNGHT1}}=0.6183\) (\(p=0.0044\)), which multiplies \(1+\phi\) by about \(e^{0.6183}\approx1.86\).

These findings have direct implications for corporate boards and regulatory policy. A firm that attracts a white knight should anticipate both more bids and greater variability in the bidding process. A firm that uses real asset restrictions should anticipate greater variability, although the mean effect of real asset restrictions on the expected number of bids is not statistically significant (\(p=0.5645\) in the constant dispersion model and \(p=0.5850\) in the varying dispersion model). Forecasts should therefore allow for a wider range of possible outcomes. A first premium which is generous is associated with fewer bids. Using one fixed dispersion (as in \cite{Huang2017CMPmu} and \cite{SaezCastillo2013HyperPoisson}) can miss this risk. Allowing dispersion to vary produces predictions that better reflect the observed heterogeneity and can support more informed planning by managers and regulators.

\begin{figure}
    \centering
    \includegraphics[width=0.75\linewidth]{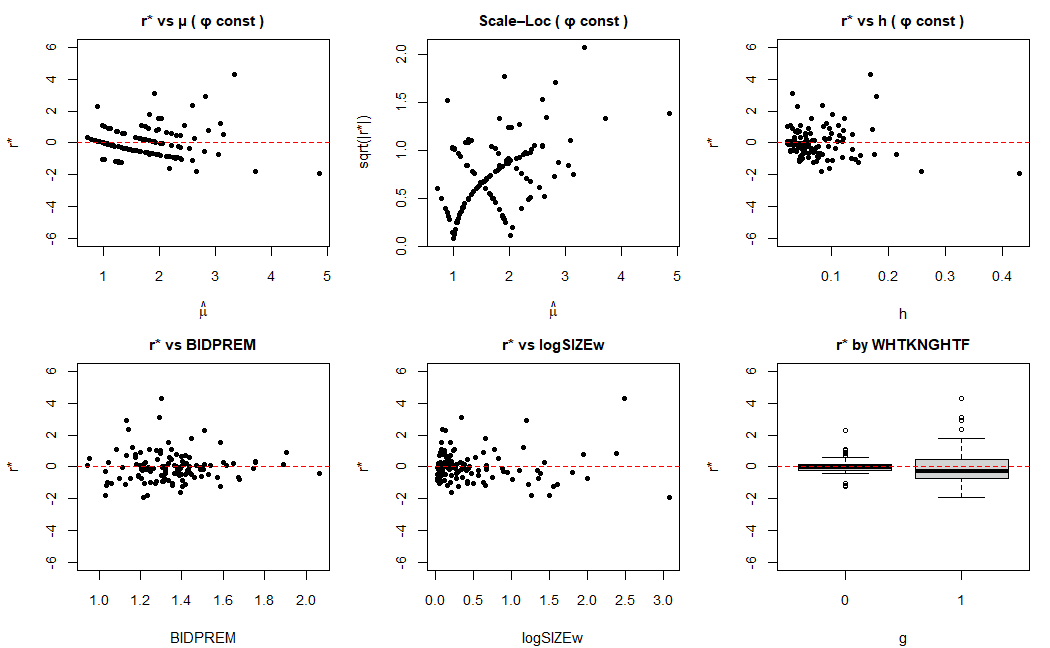}
    \caption{Takeover Bids diagnostics under constant $\phi$.}
    \label{fig:takeover_const}
\end{figure}

\begin{figure}
    \centering
    \includegraphics[width=0.75\linewidth]{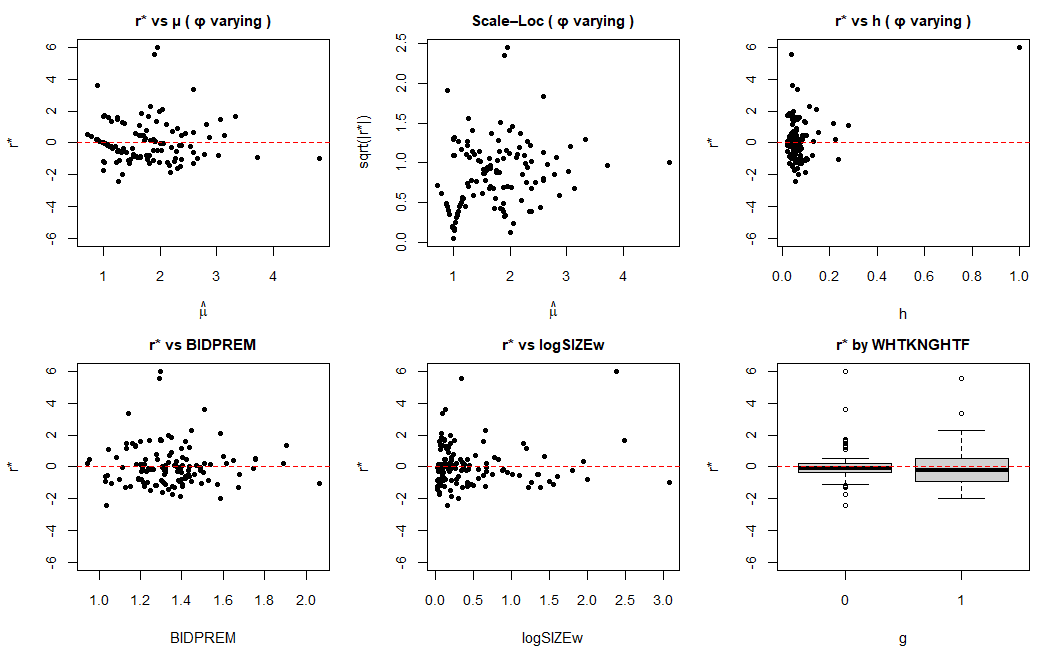}
    \caption{Takeover Bids diagnostics under varying $\phi$.}
    \label{fig:takeover_var}
\end{figure}

From Figures \ref{fig:takeover_const} and \ref{fig:takeover_var}, we learned that,
under the constant \(\phi\) fit, the top row shows signs of mismatch. The plot of \(r^{*}\) against the fitted mean \(\hat\mu\) has a light funnel shape, and the scale location panel bends upward. This says the variance grows with the mean for some firms and shrinks for others, although the model is forcing one spread for all. The plot of \(r^{*}\) against leverage \(h\) is quiet, so the lack of fit is not caused by a single point. When \(\phi\) is allowed to vary, the picture changes. The cloud of \(r^{*}\) against \(\hat\mu\) becomes even and the scale location panel is almost flat. The dispersion pattern is now absorbed by the model and the leverage plot remains calm which confirms that the improvement is global rather than driven by a few cases. The bottom row of images links residuals to key covariates. With a constant \(\phi\), there is extra spread for firms with a white knight and for firms with real asset actions, and there is a hint of widening with size. With a varying \(\phi\), these group differences largely disappear. The boxes by \texttt{WHITKNGHT} look much closer and there is no clear trend against \texttt{SIZE} or \texttt{BIDPREM}. We can therefore see that, once the model lets dispersion respond to firm actions, the residuals behave like noise. This visual check also agrees with the formal test in Table~\ref{tab:takeover_compact_named}. The dispersion slopes are jointly important with \(\chi^{2}(8)=46.1476\) and \(p<0.0001\), and the diagnostics explain why. A single spread misses the extra variability created by white knight moves and by real asset restrictions, while the varying \(\phi\) fit captures it and leaves no obvious structure in the residuals.

\subsubsection{Case Study III: Cotton Bolls Data}
For the Cotton Bolls analysis, defoliation is mean centered, with
$\texttt{defc}_i=\texttt{def}_i-\overline{\texttt{def}}$.
Thus, $\texttt{defc}=0$ corresponds to the average observed defoliation level.
The vegetative stage is used as the reference category.

\label{case3}
\begingroup
\renewcommand{\BigTblFont}{\footnotesize}
\renewcommand{\BigTblScale}{0.88}
\renewcommand{\BigTblWidth}{1.00\linewidth}
\renewcommand{\DistW}{16mm}

\begin{table}
\caption{Cotton Bolls: compact comparison of constant versus varying $\phi$ with named parameters }
\label{tab:cotton_compact_named}
\centering
{\BigTblFont
\renewcommand{\arraystretch}{1.10}
\setlength{\tabcolsep}{3pt}
\begin{adjustbox}{max width=\BigTblWidth,center}
\scalebox{\BigTblScale}{%
\begin{tabular}{|l|*{8}{P{\DistW}|}}
\hline
\multicolumn{1}{|c|}{\textbf{Param}} &
\multicolumn{4}{c|}{\textbf{Constant $\phi$}} &
\multicolumn{4}{c|}{\textbf{Varying $\phi$}} \\
\cline{2-9}
 & \textbf{Est} & \textbf{SE} & \textbf{t} & \textbf{p} & \textbf{Est} & \textbf{SE} & \textbf{t} & \textbf{p} \\
\hline
$\beta_{0}$ (Intercept) & 2.1059 & 0.0272 & 77.4084 & 0.0000 & 2.1107 & 0.0270 & 78.1174 & 0.0000 \\
$\beta_{\text{stagesbud}}$        & 0.0368 & 0.0406 &  0.9064 & 0.3647 & 0.0311 & 0.0405 &  0.7661 & 0.4436 \\
$\beta_{\text{stagesblossom}}$    & -0.2816 & 0.0518 & -5.4402 & 0.0000 & -0.2890 & 0.0518 & -5.5834 & 0.0000 \\
$\beta_{\text{stagesfig}}$        & -0.2305 & 0.0619 & -3.7252 & 0.0002 & -0.2362 & 0.0618 & -3.8190 & 0.0001 \\
$\beta_{\text{stagesboll}}$       & 0.1136 & 0.0407 &  2.7885 & 0.0053 & 0.1079 & 0.0406 &  2.6537 & 0.0080 \\
\hline
$\beta_{\text{defc}}$                 & -0.3679 & 0.0819 & -4.4900 & 0.0000 & -0.3838 & 0.0831 & -4.6180 & 0.0000 \\
$\beta_{\text{stagesbud:defc}}$       & 0.2271 & 0.1143 &  1.9866 & 0.0470 & 0.2448 & 0.1150 &  2.1283 & 0.0333 \\
$\beta_{\text{stagesblossom:defc}}$   & -0.2312 & 0.1439 & -1.6070 & 0.1080 & -0.2134 & 0.1453 & -1.4688 & 0.1419 \\
$\beta_{\text{stagesfig:defc}}$       & -0.3546 & 0.1587 & -2.2343 & 0.0255 & -0.3356 & 0.1593 & -2.1073 & 0.0351 \\
$\beta_{\text{stagesboll:defc}}$      & 0.2288 & 0.1193 &  1.9175 & 0.0552 & 0.2467 & 0.1202 &  2.0514 & 0.0402 \\
\hline
$\gamma_{0}$ (Intercept)                 & -0.0962 & 0.0048 & -19.9191 & 0.0000 & -0.0971 & 0.0079 & -12.2585 & 0.0000 \\
\hline
$\gamma_{\text{stagesbud}}$            & -- & -- & -- & -- & -0.0031 & 0.0120 & -0.2605 & 0.7945 \\
$\gamma_{\text{stagesblossom}}$        & -- & -- & -- & -- & -0.0396 & 0.0217 & -1.8271 & 0.0677 \\
$\gamma_{\text{stagesfig}}$            & -- & -- & -- & -- & 0.0345 & 0.0304 &  1.1371 & 0.2555 \\
$\gamma_{\text{stagesboll}}$           & -- & -- & -- & -- & 0.0058 & 0.0107 &  0.5437 & 0.5867 \\
\hline
$\gamma_{\text{defc}}$                  & -- & -- & -- & -- & 0.0506 & 0.0204 &  2.4823 & 0.0131 \\
$\gamma_{\text{stagesbud:defc}}$        & -- & -- & -- & -- & -0.0623 & 0.0280 & -2.2198 & 0.0264 \\
$\gamma_{\text{stagesblossom:defc}}$    & -- & -- & -- & -- & -0.2624 & 0.0911 & -2.8809 & 0.0040 \\
$\gamma_{\text{stagesfig:defc}}$        & -- & -- & -- & -- & 0.1358 & 0.0887 &  1.5297 & 0.1261 \\
$\gamma_{\text{stagesboll:defc}}$       & -- & -- & -- & -- & -0.1069 & 0.0257 & -4.1600 & 0.0000 \\
\hline
\multicolumn{9}{|p{\dimexpr\linewidth-2\fboxsep\relax}|}{
\textbf{Notes:} Constant $\phi=-0.0917$. \;
Varying $\phi(x)$: mean$=-0.0917$, min$=-0.2154$, max$=0.0311$. \;
Wald (dispersion slopes): $\chi^2(9)=32.0358$, $p=0.0002$. \;
\textbf{Runtime (proc.time):} constant (user 0.06, system 0.00, elapsed 0.06), varying (user 0.10, system 0.00, elapsed 0.11).
}\\
\hline
\end{tabular}%
}
\end{adjustbox}
} 
\end{table}
\endgroup

As we can see from Table~\ref{tab:cotton_compact_named}, at zero defoliation the expected counts by stage are high in vegetative and boll and lower in blossom and fig. The stage intercepts imply about $e^{2.1059}\!\approx\!8.2$ bolls for vegetative, $e^{2.1059+0.0368}\!\approx\!8.5$ for bud, $e^{2.1059-0.2816}\!\approx\!6.2$ for blossom, $e^{2.1059-0.2305}\!\approx\!6.5$ for fig and $e^{2.1059+0.1136}\!\approx\!9.2$ for boll. The first order defoliation slope in the baseline stage is negative with a very small $p$ value ($p<0.0001$), so more leaf removal means fewer bolls. The defoliation effect also differs by stage: bud has an additional positive defoliation interaction ($p=0.0470$ for constant $\phi$ and $p=0.0333$ for varying $\phi$), and fig has an additional negative interaction ($p=0.0255$ and $p=0.0351$). The boll interaction is borderline in the constant $\phi$ fit ($p=0.0552$) but becomes significant in the varying spread fit ($p=0.0402$). Blossom's interaction is not significant in either fit. Overall, the dominant message stays the same: defoliation is linked to fewer bolls, and the size of that drop depends on the plant stage.

An important point is that, if we impose one common dispersion level, we estimate a fixed under-dispersion of \(\hat{\phi}=-0.0917\) (from \(\gamma_0=-0.0962\)). When we let the spread change with the same covariates, the average stays near \(\hat{\phi}=-0.0917\), but the fitted values range from \(-0.2154\) to \(0.0311\), so the spread differs across stage and defoliation combinations. The dispersion slopes are important as a group \((\chi^2(9)=32.0358,\ p=0.0002)\). Some stage main effects in the dispersion are not important, for example stagesbud has \(p=0.7945\), stagesfig has \(p=0.2555\), and stagesboll has \(p=0.5867\), while stagesblossom is close but still not significant at 0.05 \((p=0.0677)\). The strongest evidence for changing spread comes from the baseline defoliation slope and the defoliation interaction terms in the dispersion part: \(\gamma_{\text{defc}}=0.0506\) \((p=0.0131)\), \(\gamma_{\text{stagesbud:defc}}=-0.0623\) \((p=0.0264)\), \(\gamma_{\text{stagesblossom:defc}}=-0.2624\) \((p=0.0040)\), and \(\gamma_{\text{stagesboll:defc}}=-0.1069\) \((p<0.0001)\), while \(\gamma_{\text{stagesfig:defc}}\) is not significant \((p=0.1261)\). This says the spread is driven mainly by defoliation and how defoliation interacts with stage, not just by stage alone.

In practical terms, the data are under-dispersed on average (more regular than Poisson) because the mean fitted \(\phi\) is negative. However, the varying dispersion fit reaches a maximum of \(\phi=0.0311\), so a small amount of over-dispersion is possible for some stage and defoliation combinations. This means yields are relatively stable around the mean on average, but the amount of that stability changes with stage and defoliation. For farm and extension stakeholders, this matters because it changes risk. If dispersion changes with stage and defoliation, then a single uncertainty band for all plots can be misleading. With the varying $\phi$ fit, you can give stage specific and defoliation level specific uncertainty, which helps planning field trials, deciding how many plants to sample at each stage, and deciding where you need more replication. For decision makers (extension teams, crop managers, and research stations), the mean effects tell how expected boll counts drop with defoliation, while the dispersion effects tell how predictable that drop is in each stage. This supports better guidance: not only ``defoliation reduces bolls'' but also ``where the outcome is more or less variable,'' which helps allocate monitoring and resources. Because earlier work often held dispersion fixed for these data \citep{Huang2017CMPmu}, our results add a useful layer that, even when the average under-dispersion looks small on the whole, the differences in spread across stage and defoliation are real and they matter.

\begin{figure}
    \centering
    \includegraphics[width=0.75\linewidth]{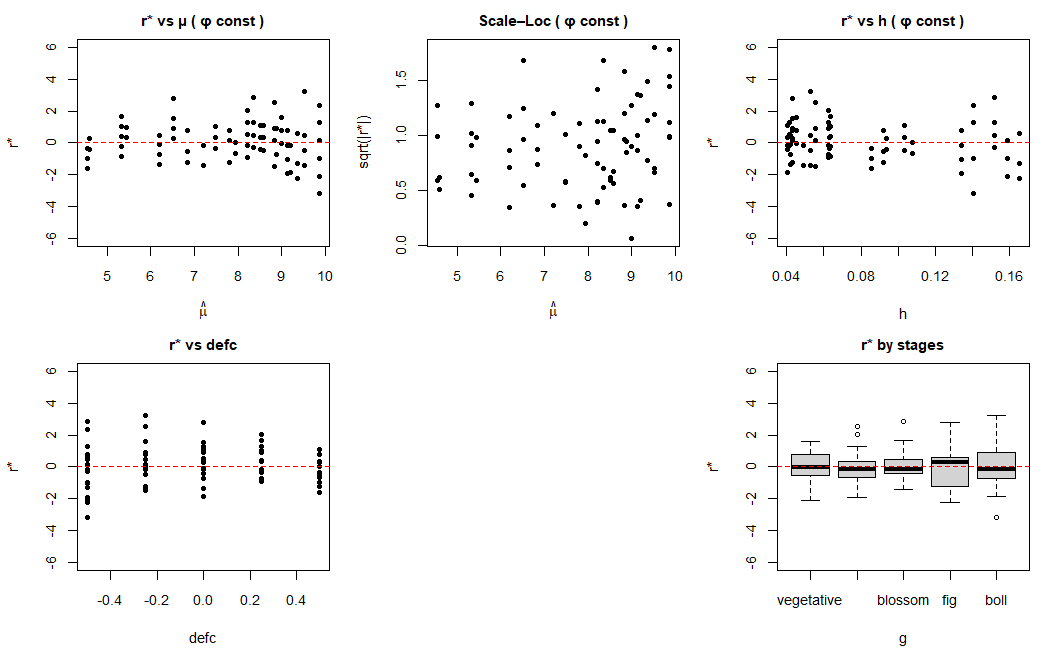}
    \caption{Cotton Bolls diagnostics under constant $\phi$.}
    \label{fig:cotton_const}
\end{figure}

\begin{figure}
    \centering
    \includegraphics[width=0.75\linewidth]{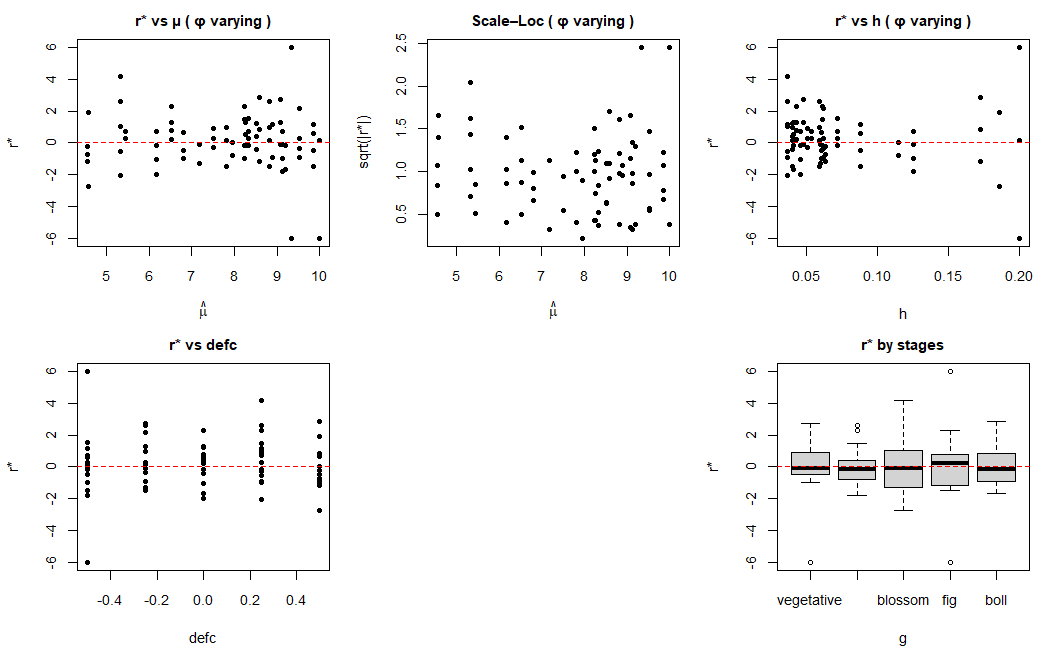}
    \caption{Cotton Bolls diagnostics under varying $\phi$.}
    \label{fig:cotton_var}
\end{figure}

When we observe Figures \ref{fig:cotton_const} and \ref{fig:cotton_var}, we can see that, with a constant \(\phi\), the top row shows small patterns. The scale–location panel bends and the cloud of \(r^{*}\) is slightly wider for some fitted means. The leverage plot is calm, so the pattern is not due to one or two points. In the bottom row, the boxplots by stage are not alike. Blossom and fig look wider and a bit off center. This hints that the spread changes across stage and defoliation combinations. When \(\phi\) is allowed to vary with the same covariates, the scale-location panel becomes flatter and the \(r^{*}\) versus fitted mean plot is more stable. The leverage plot stays calm. Most importantly, the boxplots by stage now line up around zero and have similar width. There is no clear remaining trend of \(r^{*}\) against defoliation in the varying \(\phi\) fit, which agrees with the idea that the model has absorbed much of the changing spread. This check agrees with the results in Table~\ref{tab:cotton_compact_named}. The dispersion slopes are strongly important as a group \((\chi^2(9)=32.0358,\ p=0.0002)\), and the fitted \(\phi\) values differ across stage and defoliation combinations. The fitted values range from \(-0.2154\) to \(0.0311\), so the model allows both under-dispersion and slight over-dispersion across the observed combinations, while the average fitted dispersion remains negative. A model that allows dispersion to vary captures this pattern more effectively, whereas a single constant \(\phi\) may miss important differences associated with stage and defoliation.

\section{Conclusion and Future Directions}
\label{conclu}

This paper develops a semiparametric regression framework for count data based on zero altered models and the generalized method of moments. The method allows both the conditional mean and the dispersion level to depend on covariates. It can therefore describe over-dispersion, equi-dispersion , and under-dispersion within one regression framework.

The proposed estimator uses the first two conditional factorial moments instead of a fully specified probability distribution. This reduces dependence on a particular parametric count distribution. The estimator remains consistent when the stated mean and dispersion moment conditions are correctly specified, even if other features of the full distribution are not correctly modeled. This form of robustness is an important advantage of moment based estimation \citep{Hansen1982,White1982}.

We established consistency and asymptotic normality under suitable regularity conditions. These results provide a basis for standard errors, confidence intervals, and hypothesis tests. The simulation studies showed small estimation bias and generally reasonable agreement between the Monte Carlo standard deviations and the estimated standard errors under the designs considered. The real data applications also showed that allowing dispersion to depend on covariates can reveal patterns that are not captured by a constant dispersion model.

Future research can extend the current framework in several directions. One direction is to compare the GMM objective with empirical likelihood and exponentially tilted empirical likelihood methods. These approaches can provide likelihood type inference while still using moment restrictions \citep{Owen2001,Schennach2007}. Another direction is to study improved finite sample inference for dispersion parameters, especially in settings where Wald confidence intervals have lower than expected coverage.

The framework can also be extended to higher dimensional regression settings. Penalized and Bayesian methods may help identify important predictors when the number of possible covariates is large. Possible approaches include Bayesian variable selection, the Bayesian lasso, and spike and slab methods \citep{GeorgeMcCulloch1993,ParkCasella2008,IshwaranRao2005}. Further work may also consider longitudinal data, clustered observations, and other forms of dependence.

The proposed method provides a flexible and practical way to analyze dispersed count data without requiring a complete parametric distribution. It offers a useful foundation for regression modeling when both the mean and the variability of the count outcome may change with observed covariates.

\bibliographystyle{plainnat}
\bibliography{references}



\clearpage
\appendix

\section{Derivation of the Moment Equations}
\label{app:quadratic_variance_derivation}
\label{moments_proofs}

By Lemma 1 of \citet{Ghosh2007}, if \(X\sim f_\delta\), then
\[
\mu=\mathbb E[X]=\omega\mu_0.
\]

Also,
\[
\sigma^2=\mathrm{Var}(X)
=
\omega\sigma_0^2+\frac{1-\omega}{\omega}\mu^2.
\]

Since the baseline distribution is equi-dispersed,
\(\sigma_0^2=\mu_0\). Hence,
\[
\sigma^2
=
\omega\mu_0+\frac{1-\omega}{\omega}\mu^2
=
\mu+\frac{1-\omega}{\omega}\mu^2.
\]

Therefore,
\[
\sigma^2-\mu
=
\frac{1-\omega}{\omega}\mu^2.
\]

Dividing by \(\mu^2\), define the dispersion parameter
\[
\phi
=
\frac{\sigma^2-\mu}{\mu^2}
=
\frac{1-\omega}{\omega}.
\]

Thus,
\[
\mathrm{Var}(X)=\mu+\phi\mu^2.
\]

Since \(\omega>0\), it follows that \(\phi>-1\). Also,
\(\phi<0\) corresponds to under-dispersion,
\(\phi=0\) corresponds to equi-dispersion, and
\(\phi>0\) corresponds to over-dispersion.

Using the identity
\[
\mathbb E[X(X-1)]
=
\mathrm{Var}(X)+\{\mathbb E(X)\}^2-\mathbb E(X),
\]
we obtain
\[
\begin{aligned}
\mathbb E[X(X-1)]
&=
\mu+\phi\mu^2+\mu^2-\mu\\
&=
(1+\phi)\mu^2.
\end{aligned}
\]

Hence, the first two factorial moments are
\[
\mathbb E[X]=\mu,
\qquad
\mathbb E[X(X-1)]=(1+\phi)\mu^2.
\]

For the regression setting, let
\[
\mu_i(\beta)=g(X_i^\top\beta),
\qquad
1+\phi_i(\gamma)=h(X_i^\top\gamma).
\]

Conditioning on \(X_i\),
\[
\mathbb E[Y_i\mid X_i]
=
\mu_i(\beta),
\]
and
\[
\mathbb E[Y_i(Y_i-1)\mid X_i]
=
\mu_i(\beta)^2\{1+\phi_i(\gamma)\}.
\]

Therefore, at the true parameter value
\(\theta_0=(\beta_0^\top,\gamma_0^\top)^\top\),
\[
\mathbb E\!\left[
Y_i-\mu_i(\beta_0)
\mid X_i
\right]
=0,
\]
and
\[
\mathbb E\!\left[
Y_i(Y_i-1)
-
\mu_i(\beta_0)^2
\{1+\phi_i(\gamma_0)\}
\mid X_i
\right]
=0.
\]

Using \(X_i\) as the instrument vector gives the GMM moment equations
\[
\mathbb E\!\left[
X_i\{Y_i-\mu_i(\beta_0)\}
\right]
=0,
\]
and
\[
\mathbb E\!\left[
X_i
\left\{
Y_i(Y_i-1)
-
\mu_i(\beta_0)^2
\{1+\phi_i(\gamma_0)\}
\right\}
\right]
=0.
\]

Thus, the moment vector used for GMM estimation is
\[
\psi(\mathcal U_i,\theta)
=
\begin{bmatrix}
X_i\{Y_i-\mu_i(\beta)\}\\[2mm]
X_i\left\{
Y_i(Y_i-1)
-
\mu_i(\beta)^2
\{1+\phi_i(\gamma)\}
\right\}
\end{bmatrix},
\]
with
\[
\mathbb E\{\psi(\mathcal U_i,\theta_0)\}=0.
\]
\clearpage

\section{Verification of the Moment and Envelope Conditions}
\label{app:moment_envelopes}
\label{moments_envelopes}

We explain why the moment and envelope assumptions are imposed and give sufficient conditions that guarantee them. The goal is to obtain integrable envelope bounds for the moment vector and its derivatives, uniformly over \(\theta\in\Theta\). Throughout this verification, \(\|\cdot\|_2\) denotes the Euclidean norm.

\paragraph*{Moment functions and basic bounds.}

Recall the links
\[
\mu(\beta)=g(X^\top\beta),
\qquad
1+\phi(\gamma)=h(X^\top\gamma),
\]
and the residuals
\[
r_1(\beta)=Y-\mu(\beta),
\qquad
r_2(\beta,\gamma)
=
Y(Y-1)-\mu(\beta)^2h(X^\top\gamma).
\]

Since the same covariate vector \(X\) is used as the instrument in both blocks, the moment vector is
\[
\psi(\mathcal U,\theta)
=
\begin{bmatrix}
X\,r_1(\beta)\\
X\,r_2(\beta,\gamma)
\end{bmatrix}.
\]

For any \(\beta\) and \(\gamma\),
\[
\begin{aligned}
|r_1(\beta)|
&\le |Y|+|\mu(\beta)|,\\
|r_2(\beta,\gamma)|
&\le Y^2+|Y|
+|\mu(\beta)|^2\,|h(X^\top\gamma)|,\\
\|X\,r_1(\beta)\|_2
&=
\|X\|_2|r_1(\beta)|\\
&\le
\|X\|_2\Big(|Y|+|\mu(\beta)|\Big),\\
\|X\,r_2(\beta,\gamma)\|_2
&=
\|X\|_2|r_2(\beta,\gamma)|\\
&\le
\|X\|_2
\Big(
Y^2+|Y|
+|\mu(\beta)|^2\,|h(X^\top\gamma)|
\Big).
\end{aligned}
\]

The equalities
\(
\|X\,r_1(\beta)\|_2=\|X\|_2|r_1(\beta)|
\)
and
\(
\|X\,r_2(\beta,\gamma)\|_2
=
\|X\|_2|r_2(\beta,\gamma)|
\)
hold because \(r_1(\beta)\) and \(r_2(\beta,\gamma)\) are scalars.

The envelope assumptions are chosen so that the right sides above have finite expectations uniformly over \(\theta\in\Theta\).

\paragraph*{A sufficient condition for the envelope bounds.}

\begin{lemma}[A sufficient condition for envelopes]\label{lem:env_suff}
Assume:
\begin{enumerate}
\item[(i)] \(\Theta\) is compact and there exist constants
\(B_\beta,B_\gamma<\infty\) such that
\[
\|\beta\|_2\le B_\beta,
\qquad
\|\gamma\|_2\le B_\gamma
\]
for all \(\theta\in\Theta\).

\item[(ii)] There exist constants \(C>0\) and \(B>0\) such that, for all \(t\in\mathbb R\),
\[
|g(t)|+|g'(t)|+|h(t)|+|h'(t)|
\le
C\exp(B|t|).
\]
\end{enumerate}

Then there exist nondecreasing functions
\(G_1,G_2:[0,\infty)\to[0,\infty)\) such that, for all
\(\theta\in\Theta\) and all \((X,Y)\),
\[
|\mu(\beta)|\le G_1(\|X\|_2),
\]
\[
|g(X^\top\beta)^2h(X^\top\gamma)|
\le
G_2(\|X\|_2),
\]
\[
2|g(X^\top\beta)g'(X^\top\beta)h(X^\top\gamma)|
\le
G_2(\|X\|_2),
\]
and
\[
|g(X^\top\beta)^2h'(X^\top\gamma)|
\le
G_2(\|X\|_2).
\]

One choice is
\[
G_1(u)=C\exp(BB_\beta u),
\qquad
G_2(u)=C_2\exp\big(B(2B_\beta+B_\gamma)u\big),
\]
where \(C_2<\infty\) is a constant depending on \(C\).
\end{lemma}

\begin{proof}
Fix \(\theta\in\Theta\). By the Cauchy--Schwarz inequality for the Euclidean inner product,
\[
|X^\top\beta|
\le
\|X\|_2\|\beta\|_2
\le
B_\beta\|X\|_2,
\]
and
\[
|X^\top\gamma|
\le
\|X\|_2\|\gamma\|_2
\le
B_\gamma\|X\|_2.
\]

Using the growth bound in condition (ii),
\[
\begin{aligned}
|\mu(\beta)|
&=
|g(X^\top\beta)|\\
&\le
C\exp\big(B|X^\top\beta|\big)\\
&\le
C\exp(BB_\beta\|X\|_2)\\
&=
G_1(\|X\|_2).
\end{aligned}
\]

Thus, \(G_1\) is nondecreasing. Also,
\[
\begin{aligned}
|g(X^\top\beta)^2h(X^\top\gamma)|
&\le
\left\{C\exp(B|X^\top\beta|)\right\}^2
\left\{C\exp(B|X^\top\gamma|)\right\}\\
&\le
C^3
\exp\big(B(2B_\beta+B_\gamma)\|X\|_2\big).
\end{aligned}
\]

Similarly,
\[
\begin{aligned}
2|g(X^\top\beta)g'(X^\top\beta)h(X^\top\gamma)|
&\le
2C^3
\exp\big(B(2B_\beta+B_\gamma)\|X\|_2\big),
\end{aligned}
\]
and
\[
|g(X^\top\beta)^2h'(X^\top\gamma)|
\le
C^3
\exp\big(B(2B_\beta+B_\gamma)\|X\|_2\big).
\]

The stated bounds hold after choosing \(C_2\) sufficiently large.
\end{proof}

\paragraph*{Integrability conditions.}

\begin{lemma}[Sufficient conditions for the moment envelopes]\label{lem:ass2_suff}
Assume the conditions of Lemma~\ref{lem:env_suff}. Also assume:
\begin{enumerate}
\item[(a)] There exists
\(\delta_X>0\) such that
\[
\mathbb E\|X\|_2^{2+\delta_X}
<\infty.
\]

\item[(b)]
\[
\mathbb E(Y^4)<\infty.
\]

\item[(c)]
\[
\mathbb E\left[\|X\|_2G_1(\|X\|_2)\right]<\infty
\]
and
\[
\mathbb E\left[\|X\|_2G_2(\|X\|_2)\right]<\infty.
\]
\end{enumerate}

Then
\(
\mathbb E\left[
\|X\|_2
\left\{|Y|+G_1(\|X\|_2)\right\}
\right]
<\infty,
\)
and
\(
\mathbb E\left[
\|X\|_2
\left\{Y^2+Y+G_2(\|X\|_2)\right\}
\right]
<\infty.
\)
\end{lemma}

\begin{proof}
For the first expectation, \(\mathbb E(Y^4)<\infty\) implies
\(\mathbb E(Y^2)<\infty\) and \(\mathbb E|Y|<\infty\). Hence, by the Cauchy--Schwarz inequality,
\[
\mathbb E\left[\|X\|_2|Y|\right]
\le
\left\{\mathbb E\|X\|_2^2\right\}^{1/2}
\left\{\mathbb E(Y^2)\right\}^{1/2}
<\infty,
\]
since
\(
\mathbb E\|X\|_2^2<\infty
\)
follows from
\(
\mathbb E\|X\|_2^{2+\delta_X}
<\infty.
\)

Also,
\(
\mathbb E\left[\|X\|_2G_1(\|X\|_2)\right]<\infty
\)
holds by condition (c). Therefore,
\[
\mathbb E\left[
\|X\|_2
\left\{|Y|+G_1(\|X\|_2)\right\}
\right]
<\infty.
\]

For the second expectation, \(\mathbb E(Y^4)<\infty\) implies
\(\mathbb E(Y^2)<\infty\) and
\(\mathbb E|Y|<\infty\). Hence,
\[
\mathbb E\left[\|X\|_2Y^2\right]
\le
\left\{\mathbb E\|X\|_2^2\right\}^{1/2}
\left\{\mathbb E(Y^4)\right\}^{1/2}
<\infty,
\]
and
\[
\mathbb E\left[
\|X\|_2|Y|
\right]
\le
\left\{\mathbb E\|X\|_2^2\right\}^{1/2}
\left\{\mathbb E(Y^2)\right\}^{1/2}
<\infty.
\]

Finally,
\[
\mathbb E\left[\|X\|_2G_2(\|X\|_2)\right]<\infty
\]
holds by condition (c). Combining these bounds gives
\[
\mathbb E\left[
\|X\|_2
\left\{Y^2+Y+G_2(\|X\|_2)\right\}
\right]
<\infty.
\]
\end{proof}

\paragraph*{The case of log link functions.}

If \(g(t)=\exp(t)\) and \(h(t)=\exp(t)\), then
\[
g'(t)=\exp(t)
\qquad\text{and}\qquad
h'(t)=\exp(t),
\]
so Lemma~\ref{lem:env_suff} holds with exponential envelopes. With compact \(\Theta\),
\[
G_1(u)=\exp(B_\beta u),
\qquad
G_2(u)=C_2\exp((2B_\beta+B_\gamma)u).
\]

Thus, for log links, the envelope condition is satisfied if
\[
\mathbb E\left[
\|X\|_2\exp(B_\beta\|X\|_2)
\right]
<\infty,
\]
and
\[
\mathbb E\left[
\|X\|_2
\exp((2B_\beta+B_\gamma)\|X\|_2)
\right]
<\infty,
\]
together with
\[
\mathbb E\|X\|_2^{2+\delta_X}
<\infty
\]
and
\[
\mathbb E(Y^4)<\infty.
\]

\clearpage

\section{Proofs of the Main Theoretical Results}
\label{app:proofs}

\subsection{Proof of the Jacobian Lemma}
\label{app:proof-jacobian}

\begin{proof}
Throughout this proof, derivatives of scalar residuals with respect to
\(\beta^\top\) and \(\gamma^\top\) are written as row vectors.

First note that \(\mu_i(\beta)=g(X_i^\top\beta)\), so by the chain rule,
\[
\frac{\partial \mu_i(\beta)}{\partial\beta^\top}
=
g'(X_i^\top\beta)\,X_i^\top.
\]
Thus,
\[
\frac{\partial r_{1i}(\beta)}{\partial\beta^\top}
=
-g'(X_i^\top\beta)\,X_i^\top,
\qquad
\frac{\partial r_{1i}(\beta)}{\partial\gamma^\top}=0.
\]
Hence,
\[
\frac{\partial}{\partial\beta^\top}\{X_i r_{1i}(\beta)\}
=
-X_i g'(X_i^\top\beta)X_i^\top,
\qquad
\frac{\partial}{\partial\gamma^\top}\{X_i r_{1i}(\beta)\}=0.
\]

Also,
\[
r_{2i}(\beta,\gamma)
=
Y_i(Y_i-1)
-
g(X_i^\top\beta)^2 h(X_i^\top\gamma).
\]
Therefore,
\[
\frac{\partial r_{2i}}{\partial\beta^\top}
=
-2 g(X_i^\top\beta)
g'(X_i^\top\beta)
h(X_i^\top\gamma)\,X_i^\top,
\]
and
\[
\frac{\partial r_{2i}}{\partial\gamma^\top}
=
-g(X_i^\top\beta)^2
h'(X_i^\top\gamma)\,X_i^\top.
\]
It follows that
\[
\frac{\partial}{\partial\beta^\top}
\{X_i r_{2i}(\beta,\gamma)\}
=
-2X_i
g(X_i^\top\beta)
g'(X_i^\top\beta)
h(X_i^\top\gamma)
X_i^\top,
\]
and
\[
\frac{\partial}{\partial\gamma^\top}
\{X_i r_{2i}(\beta,\gamma)\}
=
-X_i
g(X_i^\top\beta)^2
h'(X_i^\top\gamma)
X_i^\top.
\]

Taking expectations at \(\theta_0\) gives the block matrix \(G\)
stated in the main article.
\end{proof}

\subsection{Proof of the Identification Lemma}
\label{app:proof-identification}

\begin{proof}
Write \(\Psi(\theta)=\mathbb E[\psi(\mathcal U,\theta)]\) and partition \(\Psi=(\Psi_1^\top,\Psi_2^\top)^\top\), where
\[
\Psi_1(\beta)=\mathbb E[X\{Y-\mu(\beta)\}],
\qquad
\Psi_2(\beta,\gamma)=\mathbb E[X\{Y(Y-1)-\mu(\beta)^2 h(X^\top\gamma)\}].
\]
At the true parameter \(\theta_0=(\beta_0^\top,\gamma_0^\top)^\top\), the model implies
\[
\mathbb E[Y\mid X]=\mu(\beta_0),
\qquad
\mathbb E[Y(Y-1)\mid X]=\mu(\beta_0)^2 h(X^\top\gamma_0),
\]
so \(\Psi(\theta_0)=0\).

Now suppose \(\Psi(\theta)=0\). We first show that \(\beta=\beta_0\). From \(\Psi_1(\beta)=0\),
\[
0
=
\mathbb E[X\{\mu(\beta_0)-\mu(\beta)\}].
\]
Let \(A\) be an antiderivative of \(g\), so that \(A'=g\), and define
\[
L_1(\beta)
=
\mathbb E\big[A(X^\top\beta)-\mu(\beta_0)\,X^\top\beta\big].
\]
Then
\[
\nabla_\beta L_1(\beta)
=
\mathbb E[X\{\mu(\beta)-\mu(\beta_0)\}].
\]
Thus, \(\nabla_\beta L_1(\beta)=0\) is equivalent to \(\Psi_1(\beta)=0\). Also,
\[
\nabla_\beta^2 L_1(\beta)
=
\mathbb E[g'(X^\top\beta)XX^\top]\succ0.
\]
Therefore, \(L_1\) is strictly convex and has a unique stationary point. Since \(\beta_0\) solves the first order condition, it follows that \(\beta=\beta_0\).

Now fix \(\beta=\beta_0\). From \(\Psi_2(\beta_0,\gamma)=0\),
\[
0
=
\mathbb E[X\,\mu(\beta_0)^2\{h(X^\top\gamma_0)-h(X^\top\gamma)\}].
\]
Equivalently,
\[
\mathbb E[X\,\mu(\beta_0)^2\{h(X^\top\gamma)-h(X^\top\gamma_0)\}]=0.
\]
Let \(B\) be an antiderivative of \(h\), so that \(B'=h\), and define
\[
L_2(\gamma)
=
\mathbb E\big[\mu(\beta_0)^2\{B(X^\top\gamma)-h(X^\top\gamma_0)\,X^\top\gamma\}\big].
\]
Then
\[
\nabla_\gamma L_2(\gamma)
=
\mathbb E[X\,\mu(\beta_0)^2\{h(X^\top\gamma)-h(X^\top\gamma_0)\}],
\]
and
\[
\nabla_\gamma^2 L_2(\gamma)
=
\mathbb E[\mu(\beta_0)^2 h'(X^\top\gamma)XX^\top]\succ0.
\]
Therefore, \(L_2\) is strictly convex and has a unique stationary point. Since \(\gamma_0\) solves the first order condition, it follows that \(\gamma=\gamma_0\). Hence \(\theta=\theta_0\).
\end{proof}

\subsection{Proof of the Uniform Law Lemma}
\label{app:proof-ulln}

\begin{proof}
We prove
\[
\sup_{\theta\in\Theta}\|\bar\psi_n(\theta)-\Psi(\theta)\| \xrightarrow{p} 0.
\]
For each fixed \(\mathcal U\), the map \(\theta\mapsto \psi(\mathcal U,\theta)\) is continuous because \(g\) and \(h\) are continuous and \(\Theta\) is compact. Also,
\[
|r_1(\beta)|\le |Y|+|\mu(\beta)|,
\qquad
|r_2(\beta,\gamma)|\le Y^2+Y+\mu(\beta)^2h(X^\top\gamma).
\]
Using the envelope functions from the moment and envelope assumptions,
\[
\sup_{\theta\in\Theta}\|\psi(\mathcal U,\theta)\|
\le
\|X\|\{|Y|+G_1(\|X\|)\}
+
\|X\|\{Y^2+Y+G_2(\|X\|)\}.
\]
The expectation of the right side is finite. Since \(\Theta\) is compact, \(\theta\mapsto \psi(\mathcal U,\theta)\) is continuous, and the class is dominated by an integrable envelope, a standard uniform law of large numbers gives
\[
\sup_{\theta\in\Theta}\|\bar\psi_n(\theta)-\Psi(\theta)\| \xrightarrow{p} 0.
\]
\citep[Thm.~2.4]{NeweyMcFadden1994}.
\end{proof}

\subsection{Proof of the Consistency Theorem}
\label{app:proof-consistency}

\begin{proof}
Define
\[
Q_0(\theta)=\Psi(\theta)^\top W_0 \Psi(\theta).
\]
By the identification lemma in the main article, \(\Psi(\theta)=0\) implies \(\theta=\theta_0\), so \(Q_0(\theta)\ge 0\) with unique minimizer \(\theta_0\).

We show
\[
\sup_{\theta\in\Theta}|Q_n(\theta)-Q_0(\theta)| \xrightarrow{p} 0.
\]
Using \(W_n\xrightarrow{p}W_0\) and the uniform law lemma in the main article,
\[
\sup_{\theta\in\Theta}|Q_n(\theta)-Q_0(\theta)|
\xrightarrow{p}
0.
\]
The argmin theorem then gives
\[
\hat\theta_n\xrightarrow{p}\theta_0.
\]
\citep[Thm.~5.7]{vanDerVaart1998}.
\end{proof}

\subsection{Proof of the Linearization Lemma}
\label{app:proof-linearization}

\begin{proof}
Let
\[
D_n(\theta)=\frac{\partial \bar\psi_n(\theta)}{\partial\theta^\top}.
\]
Because \(\hat\theta_n\) minimizes
\[
Q_n(\theta)=\bar\psi_n(\theta)^\top W_n\bar\psi_n(\theta),
\]
a first order condition gives
\[
0
=
\nabla_\theta Q_n(\hat\theta_n)
=
2D_n(\hat\theta_n)^\top W_n \bar\psi_n(\hat\theta_n).
\]
By a mean value expansion of \(\bar\psi_n(\hat\theta_n)\) around \(\theta_0\), for some \(\tilde\theta_n\) on the line segment between \(\hat\theta_n\) and \(\theta_0\),
\[
\bar\psi_n(\hat\theta_n)
=
\bar\psi_n(\theta_0)+D_n(\tilde\theta_n)(\hat\theta_n-\theta_0).
\]
Substituting into the first order condition gives
\[
0
=
2D_n(\hat\theta_n)^\top W_n\bar\psi_n(\theta_0)
+
2D_n(\hat\theta_n)^\top W_n D_n(\tilde\theta_n)(\hat\theta_n-\theta_0).
\]
Solving for \(\sqrt n(\hat\theta_n-\theta_0)\) gives
\[
\sqrt n(\hat\theta_n-\theta_0)
=
-
\Big(D_n(\hat\theta_n)^\top W_n D_n(\tilde\theta_n)\Big)^{-1}
D_n(\hat\theta_n)^\top W_n \sqrt n\,\bar\psi_n(\theta_0).
\]
By consistency, derivative convergence, and \(W_n\xrightarrow{p}W_0\),
\[
D_n(\hat\theta_n)\xrightarrow{p}G,
\qquad
D_n(\tilde\theta_n)\xrightarrow{p}G,
\qquad
W_n\xrightarrow{p}W_0.
\]
By the full rank lemma in the main article, \(G^\top W_0G\) is invertible. Therefore,
\[
\sqrt n(\hat\theta_n-\theta_0)
=
-(G^\top W_0G)^{-1}G^\top W_0 \sqrt n\,\bar\psi_n(\theta_0)+o_p(1).
\]
\end{proof}

\subsection{Proof of the Full Rank Lemma}
\label{app:proof-fullrank}

\begin{proof}
Recall the block form
\[
G
=
-\begin{bmatrix}
G_{11} & 0\\
G_{21} & G_{22}
\end{bmatrix},
\]
where
\[
G_{11}=\mathbb E[X\,g'(X^\top\beta_0)X^\top],
\qquad
G_{22}=\mathbb E[X\,\mu(\beta_0)^2 h'(X^\top\gamma_0)X^\top].
\]
For any nonzero vector \(v\in\mathbb R^p\),
\[
v^\top G_{11} v
=
\mathbb E[g'(X^\top\beta_0)(v^\top X)^2]>0,
\]
because \(g'(X^\top\beta_0)>0\) almost surely and \(\mathbb E[XX^\top]\succ0\). Thus \(G_{11}\succ0\). Similarly, for any nonzero \(w\in\mathbb R^p\),
\[
w^\top G_{22} w
=
\mathbb E[\mu(\beta_0)^2 h'(X^\top\gamma_0)(w^\top X)^2]>0.
\]
Thus \(G_{22}\succ0\). Since both diagonal blocks are full rank and \(G\) is block lower triangular, \(G\) has full column rank \(2p\).
\end{proof}

\subsection{Proof of the Asymptotic Normality Theorem}
\label{app:proof-an}

\begin{proof}
By the linearization lemma in the main article,
\[
\sqrt n(\hat\theta_n-\theta_0)
=
-(G^\top W_0G)^{-1}G^\top W_0 \sqrt n\,\bar\psi_n(\theta_0)+o_p(1).
\]
It remains to show a multivariate central limit theorem for \(\sqrt n\,\bar\psi_n(\theta_0)\). Since \(\{\mathcal U_i\}\) are i.i.d., the random vectors \(\psi(\mathcal U_i,\theta_0)\) are i.i.d.\ with mean zero. The moment assumptions imply
\[
\mathbb E\|\psi(\mathcal U,\theta_0)\|^2<\infty.
\]
Therefore, by the multivariate central limit theorem,
\[
\sqrt n\,\bar\psi_n(\theta_0)\Rightarrow \mathcal N(0,S_0),
\qquad
S_0=\mathbb E[\psi(\mathcal U,\theta_0)\psi(\mathcal U,\theta_0)^\top].
\]
Applying Slutsky's theorem to the linearization gives
\[
\sqrt n(\hat\theta_n-\theta_0)\Rightarrow \mathcal N(0,V),
\]
where
\[
V=(G^\top W_0G)^{-1}G^\top W_0 S_0 W_0G(G^\top W_0G)^{-1}.
\]
If \(W_0=S_0^{-1}\), then
\[
V=(G^\top S_0^{-1}G)^{-1}.
\]
\end{proof}

\subsection{Proof of the Covariance Consistency Lemma}
\label{app:proof-shat}

\begin{proof}
Write
\[
\widehat S-S_0
=
\left(\widehat S-\frac1n\sum_{i=1}^n \psi(\mathcal U_i,\theta_0)\psi(\mathcal U_i,\theta_0)^\top\right)
+
\left(\frac1n\sum_{i=1}^n \psi(\mathcal U_i,\theta_0)\psi(\mathcal U_i,\theta_0)^\top-S_0\right).
\]
The second term goes to zero in probability by the law of large numbers. For the first term, continuity of \(\psi(\mathcal U,\theta)\) in \(\theta\), the envelope conditions, and \(\hat\theta_n\xrightarrow{p}\theta_0\) imply that the difference between evaluating at \(\hat\theta_n\) and \(\theta_0\) vanishes in probability. Hence
\[
\widehat S\xrightarrow{p}S_0.
\]
\end{proof}

\subsection{Proof of the Weighting Matrix Consistency Lemma}
\label{app:proof-what}

\begin{proof}
By the covariance consistency lemma in the main article, \(\widehat S\xrightarrow{p}S_0\). Since inversion is continuous on the set of positive definite matrices,
\[
\widehat S^{-1}\xrightarrow{p}S_0^{-1}.
\]
\end{proof}

\subsection{Proof of the Plug-in Covariance Consistency Proposition}
\label{app:proof-vhat}

\begin{proof}
By the derivative uniform law used in the linearization argument,
\[
\hat G\xrightarrow{p}G.
\]
Also, \(\widehat S\xrightarrow{p}S_0\) and \(\hat W\xrightarrow{p}W_0\). The result follows by the continuous mapping theorem. For the optimal two-step choice, \(\hat W=\widehat S^{-1}\xrightarrow{p}S_0^{-1}\).
\end{proof}

\subsection{Proof of the J Test Proposition}
\label{app:proof-j}

\begin{proof}
In a just identified GMM problem, the number of moment conditions equals the number of parameters. Hence the overidentification statistic has zero degrees of freedom. Therefore, \(J_n\) is asymptotically degenerate at zero and does not provide a nontrivial specification test.
\end{proof}

\subsection{Proof of the Wald Test Proposition}
\label{app:proof-wald}

\begin{proof}
Under \(H_0\), \(R\theta_0=r\), and by asymptotic normality,
\[
\sqrt n(R\hat\theta_n-r)
=
R\sqrt n(\hat\theta_n-\theta_0)
\Rightarrow
\mathcal N(0,RVR^\top).
\]
By the plug-in covariance consistency proposition in the main article,
\[
\widehat V\xrightarrow{p}V.
\]
Therefore,
\[
\mathcal W_n
=
(R\hat\theta_n-r)^\top
\left(
R\frac{\widehat V}{n}R^\top
\right)^{-1}
(R\hat\theta_n-r)
\xrightarrow{d}
\chi^2_{\operatorname{rank}(R)}.
\]
\end{proof}

\clearpage
\setcounter{section}{3}
\setcounter{subsection}{0}

\section{Notation, Symbols, and Abbreviations}
\label{app:notation}

\subsection{Dimensions and parameters}

\begin{description}

\item[$n$]
Sample size.

\item[$p$]
Number of regressors in the conditional mean design vector \(X_i\).

\item[$q$]
Number of regressors in the dispersion design vector \(Z_{\phi i}\).

\item[$X_i\in\mathbb R^p$]
The \(p\times1\) covariate vector for observation \(i\), used in the conditional mean model and as the instrument vector for the first moment block.

\item[$Z_{\phi i}\in\mathbb R^q$]
The \(q\times1\) covariate vector for observation \(i\), used in the dispersion model and as the instrument vector for the second moment block.

\item[$\beta\in\mathbb R^p$]
Regression parameter vector for the conditional mean model.

\item[$\gamma\in\mathbb R^q$]
Regression parameter vector for the dispersion model.

\item[$\theta=(\beta^\top,\gamma^\top)^\top\in\mathbb R^{p+q}$]
Stacked parameter vector.

\item[$\theta_0=(\beta_0^\top,\gamma_0^\top)^\top\in\mathbb R^{p+q}$]
True parameter value.

\item[$\Theta$]
Parameter space, assumed compact with
\(\theta_0\in\operatorname{int}(\Theta)\).

\item[$s$]
Number of moment conditions. In the present just identified setup,
\[
s=p+q.
\]

\item[$d$]
Number of parameters. In the present setup,
\[
d=p+q.
\]

\end{description}

For the constant dispersion model, \(q=1\), so
\[
s=d=p+1.
\]

For the varying dispersion model with \(Z_{\phi i}=X_i\), \(q=p\), so
\[
s=d=2p.
\]
\clearpage

\section{Two-Step Algorithm}
\label{app:two_step_algorithm}

\subsection*{Two-step GMM procedure}

Given \(\{\mathcal U_i\}_{i=1}^n\),
\(\psi(\mathcal U_i,\theta)\in\mathbb R^{p+q}\), and
\(\Theta\subset\mathbb R^{p+q}\), define
\[
\bar\psi_n(\theta)
=
\frac{1}{n}\sum_{i=1}^n
\psi(\mathcal U_i,\theta),
\qquad
Q_n(\theta,W)
=
\bar\psi_n(\theta)^\top
W
\bar\psi_n(\theta).
\]

\paragraph*{Step 1: Identity-weighted estimation.}
Set
\[
W^{(1)}=I_{p+q}
\]
and compute
\[
\hat\theta^{(1)}
=
\arg\min_{\theta\in\Theta}
Q_n(\theta,W^{(1)}).
\]

\paragraph*{Step 2: Optimal weight update.}
Set
\[
\bar\psi_n^{(1)}
=
\bar\psi_n(\hat\theta^{(1)}),
\]
and compute
\[
\widehat S^{(1)}
=
\frac{1}{n}
\sum_{i=1}^n
\left\{
\psi(\mathcal U_i,\hat\theta^{(1)})
-
\bar\psi_n^{(1)}
\right\}
\left\{
\psi(\mathcal U_i,\hat\theta^{(1)})
-
\bar\psi_n^{(1)}
\right\}^{\top}.
\]
Set
\[
W^{(2)}
=
\left(\widehat S^{(1)}\right)^{-1}.
\]

\paragraph*{Step 3: Second-step estimation.}
Compute
\[
\hat\theta^{(2)}
=
\arg\min_{\theta\in\Theta}
Q_n(\theta,W^{(2)}).
\]

\paragraph*{Step 4: Covariance estimation.}
Set
\[
\bar\psi_n^{(2)}
=
\bar\psi_n(\hat\theta^{(2)}),
\]
and compute
\[
\widehat S^{(2)}
=
\frac{1}{n}
\sum_{i=1}^n
\left\{
\psi(\mathcal U_i,\hat\theta^{(2)})
-
\bar\psi_n^{(2)}
\right\}
\left\{
\psi(\mathcal U_i,\hat\theta^{(2)})
-
\bar\psi_n^{(2)}
\right\}^{\top}.
\]

Compute the sample Jacobian
\[
\widehat G^{(2)}
=
\left.
\frac{\partial\bar\psi_n(\theta)}
{\partial\theta^\top}
\right|_{\theta=\hat\theta^{(2)}},
\]
and define
\[
\widehat A
=
\widehat G^{(2)\top}
W^{(2)}
\widehat G^{(2)}.
\]

Estimate
\[
\widehat{\mathrm{Var}}(\hat\theta^{(2)})
=
\frac{1}{n}
\widehat A^{-1}
\left\{
\widehat G^{(2)\top}
W^{(2)}
\widehat S^{(2)}
W^{(2)}
\widehat G^{(2)}
\right\}
\widehat A^{-1}.
\]

The procedure returns
\(\hat\theta^{(2)}\) and
\(\widehat{\mathrm{Var}}(\hat\theta^{(2)})\).

\clearpage

\section{Data generation for simulation study}
\label{app:simulation_dgps}

This appendix records the probability models used to generate the simulated
data discussed in Section 4 of the main article. Except where noted otherwise,
the covariate is generated as
\[
x_i\sim\mathrm{Uniform}(-1,1),
\]
and the target conditional mean and dispersion satisfy
\[
\mu_i=\exp(\beta_0+\beta_1x_i),
\qquad
1+\phi_i=\exp(\gamma_0+\gamma_1x_i).
\]
The parameters of each baseline count distribution are then combined with the
zero-alteration mechanism so that the generated responses follow the intended
first two conditional factorial moments. The zero-altered constructions use
the framework of \citet{Ghosh2007}.

\subsection{ZA Binomial}

The first design starts with a Binomial baseline having \(N=75\) trials and
success probability \(p_i\). Its probability mass function is
\[
f_0(y_i)
=
\binom{N}{y_i}
p_i^{y_i}(1-p_i)^{N-y_i},
\qquad
y_i=0,\ldots,N,
\]
with baseline zero probability
\[
\pi_{0i}=f_0(0)=(1-p_i)^N.
\]

For each observation, the quantities \(\omega_i\), \(\pi_{0i}\), and
\(\delta_i\) determine the corresponding zero-altered distribution. They are
chosen within the ZA construction to produce the required conditional mean
and dispersion. The regression coefficients for this design are
\[
(\beta_0,\beta_1,\gamma_0,\gamma_1)
=
(\log(1.5),0.10,-0.20,-0.10).
\]

\subsection{ZA COM--Poisson}

A second design uses the COM--Poisson distribution
\citep{ConwayMaxwell1962,Shmueli2005}, whose probability mass function is
\[
f(y_i\mid\lambda_i,\nu_i)
=
\frac{\lambda_i^{y_i}}
{(y_i!)^{\nu_i}Z(\lambda_i,\nu_i)},
\qquad
y_i=0,1,2,\ldots,
\]
where
\[
Z(\lambda_i,\nu_i)
=
\sum_{j=0}^{\infty}
\frac{\lambda_i^j}{(j!)^{\nu_i}}
\]
is the normalizing constant.

The baseline COM--Poisson moments are determined by
\((\lambda_i,\nu_i)\), after which the zero-alteration step is applied using
the construction in \citet{Ghosh2007}. The coefficient values considered in
this experiment are
\[
(\beta_0,\beta_1,\gamma_0,\gamma_1)
=
(1.20,0.50,0,0).
\]

\subsection{Zero-Inflated Poisson}

To represent a zero-inflated setting, let \(\rho_i\) denote the probability
of a structural zero. The response distribution is
\[
P(Y_i=0)
=
\rho_i+(1-\rho_i)e^{-\lambda_i},
\]
and, for \(y_i\geq1\),
\[
P(Y_i=y_i)
=
(1-\rho_i)
\frac{e^{-\lambda_i}\lambda_i^{y_i}}{y_i!}.
\]

The zero-inflation probability and Poisson mean are parameterized as
\[
\rho_i=\frac{\phi_i}{1+\phi_i},
\qquad
\lambda_i=\frac{\mu_i}{1-\rho_i},
\]
which gives the required conditional mean and second factorial moment. We set
\[
(\beta_0,\beta_1,\gamma_0,\gamma_1)
=
(\log(3),0.40,0.20,0.20).
\]

\subsection{ZA Negative Binomial 2}

The negative-binomial experiment is based on the NB2 parameterization
\citep{CameronTrivedi1998}. Let \(m_i\) be the baseline mean and
\(\alpha_i>0\) the NB2 dispersion parameter. The baseline probability mass
function is
\[
f_0(y_i)
=
\frac{\Gamma\left(y_i+\alpha_i^{-1}\right)}
{\Gamma\left(\alpha_i^{-1}\right)y_i!}
\left(
\frac{1}{1+\alpha_i m_i}
\right)^{\alpha_i^{-1}}
\left(
\frac{\alpha_i m_i}{1+\alpha_i m_i}
\right)^{y_i},
\qquad
y_i=0,1,2,\ldots.
\]

This parameterization has
\[
\E(Y_i)=m_i,
\qquad
\Var(Y_i)=m_i+\alpha_i m_i^2,
\]
and baseline zero probability
\[
\pi_{0i}
=
\left(1+\alpha_i m_i\right)^{-\alpha_i^{-1}}.
\]

The NB2 baseline is combined with the zero-alteration mechanism described in
the main article, using the baseline zero probability and the corresponding
ZA weights to obtain the desired moment structure. The coefficient vector is
\[
(\beta_0,\beta_1,\gamma_0,\gamma_1)
=
(1.20,0.50,0,0).
\]

\subsection{ZA Poisson--Inverse Gaussian}

The final design uses a Poisson--inverse Gaussian (PIG) baseline
\citep{Willmot1987}. Unlike the preceding designs, the covariate is generated
from a standard normal distribution,
\[
x_i\sim N(0,1),
\]
and the target regression functions are
\[
\mu_i=\exp(0.60+0.40x_i),
\qquad
1+\phi_i=\exp(0.15+0.05x_i).
\]

When \(\phi_i>0\), the PIG baseline is generated hierarchically as
\[
\Lambda_i \sim \operatorname{IG}
\left(\mu_i,\frac{\mu_i}{\phi_i}\right),
\qquad
Y_i\mid\Lambda_i\sim\operatorname{Poisson}(\Lambda_i),
\]
where the second argument of the inverse Gaussian distribution is the shape
parameter. Its baseline probability of zero is
\[
\pi_i
=
\exp\left\{
\frac{1}{\phi_i}
\left(
1-\sqrt{1+2\mu_i\phi_i}
\right)
\right\}.
\]

The zero-alteration quantities are
\[
\delta_i^2=\frac{\phi_i}{1+\phi_i},
\qquad
f_{0i}(0)=(1+\phi_i)\pi_i-\phi_i,
\]
which give
\[
P(Y_i=0)
=
\delta_i^2+
(1-\delta_i^2)f_{0i}(0).
\]
If the observation is positive, it is drawn from the PIG baseline conditional
on being greater than zero.

The regression coefficients are
\[
(\beta_0,\beta_1,\gamma_0,\gamma_1)
=
(0.60,0.40,0.15,0.05).
\]
For this experiment, each simulated data set contains \(n=1000\)
observations; \(200\) Monte Carlo replications are used, with random seed
\(2025\).

\end{document}